\documentclass[a4paper,12pt]{article}

\usepackage{geometry,amssymb,amsmath,latexsym,amsthm,bm}
\usepackage[utf8]{inputenc} 
\usepackage[T1]{fontenc}
\usepackage{setspace} 
\usepackage{float}          

\usepackage{empheq} 
\usepackage[normalem]{ulem} 
\usepackage{enumitem} 
\usepackage{multirow} 
\usepackage{bbm} 
\usepackage{cases} 
\usepackage[final]{pdfpages} 
\usepackage{bigints} 
\usepackage{comment} 
\usepackage{booktabs} 
\usepackage{eurosym}
\usepackage[font=small,labelfont=bf]{caption}
\usepackage{indentfirst} 
\usepackage{breqn} 
\usepackage{cleveref}
\usepackage{appendix}
\usepackage{booktabs}       
\usepackage{multirow}
\usepackage{algorithm}  
\usepackage{algorithmic}  
\usepackage{multicol}
\usepackage{pgfplots}
\pgfplotsset{compat=1.18} 
\usepackage{enumitem} 

\usepackage[style=apa6, 
            backend=biber,
            natbib=true, 
            url=false,
            date=year,
            doi=false,
            isbn=false]{biblatex}
\DeclareLanguageMapping{american}{american-apa}
\AtEveryBibitem{\clearfield{month}} 

\usepackage{tikz}
\usetikzlibrary{snakes}
\usetikzlibrary{calc} 
\usetikzlibrary{shapes,arrows} 
\usetikzlibrary{arrows,positioning} 

\usepgflibrary{arrows} 
\usepgflibrary[arrows] 
\usetikzlibrary{arrows} 
\usetikzlibrary[arrows] 

\usepackage{graphicx} 
\usepackage{graphics}

\newtheorem{definition}{Definition}

\newtheorem{lemma}{Lemma}

\newtheoremstyle{exampstyle}
  {} 
  {} 
  {} 
  {} 
  {\bfseries\color{blue}} 
  {.} 
  {.5em} 
  {} 
\theoremstyle{exampstyle} 
\newtheorem{rmk}{Remark}
\newtheorem{thm}{Theorem}

\newtheorem{coro}{Corollary}        
\newtheorem{prop}{Proposition}

\theoremstyle{definition} 

\newtheorem{remarktmp}{Remark}
\newenvironment{remark}
{ \begin{remarktmp} 	}
{ 		\medskip\hfill{\LARGE$\lrcorner$}
		\end{remarktmp} } 
\allowdisplaybreaks[1]

\newtheorem{exampletmp}{Example}
\allowdisplaybreaks[1]

\newcommand{\BB}{\mathbf{B}} 
 
\newcommand{\BI}{\mathbf{I}}

\newcommand{\Udelta}{\bar{\delta}_\succ}
\newcommand{\Ugamma}{\bar{\gamma}_\succ}
\newcommand{\Urho}{\bar{\rho}_\succ}

\newcommand{\Lmu}{\underline{\mu}_\succ}
\newcommand{\Ldelta}{\underline{\delta}_\succ}
\newcommand{\Lgamma}{\underline{\gamma}_\succ}
\newcommand{\Lrho}{\underline{\rho}_\succ}

\newcommand{\muS}{\mu_\succ}
\newcommand{\deltaS}{\delta_\succ}
\newcommand{\gammaS}{\gamma_\succ}

\usetikzlibrary{arrows.meta}

\title{Revealed Attentional Interference\thanks{We thank Chris Chambers,  Paola Manzini, Marco Mariotti, Yusufcan Masatlioglu, Hiroki Nishimura, Fernando Payro, Evan Piermont, Tri Phu Vu and Siming Ye, as well as seminar participants at BSE Summer Forum, UC Riverside and the DC Theory Workshop, for helpful comments. Lin gratefully acknowledges financial support from Academia Sinica and Taiwan's National Science and Technology Council.}}
\author{Paul H.Y. Cheung\thanks{Naveen Jindal School of Management, The University of Texas at Dallas. Email: paul.cheung@utdallas.edu} \and Yi-Hsuan Lin\thanks{Institute of Economics, Academia Sinica. Email: yihsuanl@econ.sinica.edu.tw} \and Chung-Hao Sheu\thanks{Email: elbert4529@gmail.com}}
\date{July 15, 2026}

\begin{document}

\maketitle

\begin{abstract}

We study the impact of external stimuli on attention in the Attentional Interference Model, capturing two opposing forces in consideration-set formation: proactive and retroactive interference. Proactive interference limits the permeation of external information, while retroactive interference displaces internally generated considerations. We model these forces using parameters governing  permeation and displacement. In a general setting, we characterize the tight range of these parameters and show that, across several specifications, only an upper bound on permeation is revealed. Imposing monotonicity on internal attention in one case both tightens this upper bound and yields a lower bound. We illustrate our results through simulation.\\
\end{abstract}

\begin{quote}
\textbf{Keywords:} Stochastic Choice, Limited Consideration, Attentional Interference, Revealed Preference

\textbf{JEL:}  D11, D80, D90
\end{quote}


\section{Introduction}

People’s attention is often guided by external stimuli, or attentional cues, that are pervasive in everyday environments. These cues, whether intentional or incidental, can significantly influence decision making. For example, salient displays or recommendations may draw attention to particular alternatives, while observing others’ behavior can expand the set of perceived options. Because an alternative must be considered to be chosen, external stimuli play a central role in shaping which options enter the decision process.\footnote{Such cues arise in a variety of settings, including retail environments (e.g., window displays, \cite{SEN2002277}; mobile in-store advertising, \cite{Mirja_et_al_2017}; end-of-aisle promotions, \cite{PHILLIPS2015141}) and social contexts (e.g., influencer marketing, electronic word-of-mouth, and reference group effects; \cite{Boerman_et_al_2022, daugherty2016ewom, Bearden_et_al_1982}).}

Understanding how external forces affect attention is therefore of first-order importance. Marketers may wish to understand how different advertising formats influence the formation of consideration sets, while policymakers may seek to disentangle intrinsic preferences from the effects of attentional manipulation. This is particularly relevant in the context of libertarian paternalism, where individuals remain free to choose, yet policymakers aim to guide which alternatives are considered. In an environment characterized by scarce attention and abundant alternatives, the ability to direct attention—without restricting choice—becomes a key determinant of economic and behavioral outcomes.

A central empirical challenge in this context is that attention is typically unobserved and costly to measure directly. Although richer measures of attention may be available in some settings, such as laboratory experiments, eye-tracking studies, or clickstream data, these measures are often unavailable, costly to collect, or difficult to link to choice behavior at scale. Moreover, even when exposure to stimuli is observed, the degree to which decision makers attend to those stimuli is inherently subjective and rarely directly measurable. Instead, researchers and practitioners often observe only coarse data, such as choices and the frequency with which certain stimuli or alternatives appear. This raises a fundamental question: to what extent can we infer the impact of external stimuli on attention and choice using only such limited, observational data?

In this paper, we attempt to open the black box of attention formation under external influence. A growing choice-theoretical literature on attention (e.g., \citealp{Manzini2014, Brady2016Menu-DependentFeasibility, Masatlioglu2012RevealedAttention, Lleras2017WhenConsideration}) has introduced unobserved consideration sets into choice models and substantially advanced our understanding of how attention allocation and preferences can be inferred from observed choices. Nevertheless, existing models typically abstract from external stimuli or treat attention as internally determined, leaving open how external forces affect attention in observational environments. To address this gap, we develop a model of stochastic consideration with external stimuli and show how their impact on attention can be inferred from choice data alone.

Specifically, drawing on insights from the psychological literature, we incorporate two opposing forces that arise in the presence of external stimuli: proactive interference and retroactive interference (e.g., \citealp{Underwood1957, MurnaneShiffrin1991, BrunelNelson2003, BurkeSrull1988, KlieglBauml2021, MurphyCastel2022}). Proactive interference refers to the tendency of previously acquired information to impede the encoding or retrieval of subsequently presented information, whereas retroactive interference refers to the tendency of newly acquired information to reduce the accessibility of previously learned information. Translated into the language of consideration-set formation, proactive interference implies that an internal attention rule may prevent external stimuli from entering the consideration set, whereas retroactive interference implies that, conditional on external stimuli being considered, alternatives that would otherwise be considered may be displaced. This distinction matters because the two forces have different welfare implications and distort revealed-preference inference in different ways. Under proactive interference, failure to choose a presented alternative need not imply that it is less preferred, since the alternative may never enter consideration. Under retroactive interference, choosing an externally induced alternative need not reveal preference over alternatives that were displaced from consideration.

In the choice-theoretic literature, the (internal) attention rule is a well-studied object, as in standard attention models (e.g., \citealp{Manzini2014, Brady2016Menu-DependentFeasibility, Masatlioglu2012RevealedAttention, Lleras2017WhenConsideration}). External stimuli have also been examined in various forms. For example, \citet{Natenzon2019}, \citet{GUNEY2018935}, and \citet{LiTangZhang2025Associative} study how exposure to non-choosable options affects choices; \citet{Cheung:2023} and \citet{Ke2021LearningBox} analyze how recommendations influence choices and beliefs; and \citet{Kashaev2023} develops a dynamic model in which a decision maker’s attention and preferences are shaped by peers’ behavior. To the best of our knowledge, this paper is the first choice-theoretic study to explicitly analyze the roles of proactive and retroactive interference in the formation of attention rules.

In \Cref{Sect:Model}, we introduce the framework and formally define the Attention Interference Model. We represent the unobservable internal attention rule and the observable external stimuli using probability distributions over subsets of a menu. The internal attention rule captures the DM’s latent consideration process, reflecting personal factors such as background, past experience, and habitual attention patterns. It corresponds to the standard unobservable consideration mechanism in stochastic choice models. In contrast, the external stimuli describe the frequency with which sets of alternatives are presented to the decision maker (henceforth, DM). The advantage of this modeling choice is that this information is objective and directly observable. Given the external stimuli, proactive and retroactive interference are modeled through two parameters in a two-stage process. First, with probability $\beta$—the \emph{permeation parameter}—external stimuli successfully enter the consideration set, overcoming proactive interference. Conditional on successful entry, with probability $\alpha$—the \emph{displacement parameter}—the external stimuli displace the internally generated consideration set, reflecting retroactive interference. Finally, the DM chooses the most preferred alternative according to their underlying preferences.

In \Cref{Sect:RevealAttInf}, we first define the object of interest: the revealed range of the underlying parameters. We then show that this range can be succinctly characterized by an inequality involving the choice rule and the external stimuli. Using the best and worst alternatives in each set, we derive intuitive bounds on $\alpha$ and $\beta$. We also introduce three benchmark cases of the model: Perfect Displacement, where $\alpha=1$; Perfect Retention, where $\alpha=0$; and Unified Attention Interference, where $\alpha=\beta$. We show that, in these cases, only the maximal degree of permeation—an upper bound on $\beta$—is revealed. We further show that these bounds are ordered monotonically, so that Perfect Displacement provides the most conservative (yet strongest) permeation bound over the three cases.

These bounds provide marketers and policymakers with a way to assess the maximal extent to which external stimuli permeate consideration sets. When only external stimuli and choices are observed, the bounds discipline how large the effect of those stimuli on consideration can be. For example, if a promoted alternative is frequently displayed but rarely chosen when it would be chosen upon consideration, the model implies that the stimulus cannot have permeated attention very often. Thus, even without directly observing consideration sets, the analyst can use choice data to rule out overly strong claims about the attentional impact of advertisements, recommendations, or other forms of external exposure.

A natural question is whether one can also identify a lower bound for the permeation parameter. We address this question in \Cref{Sect:NonPara}. We impose a non-parametric restriction on the internal attention rule: the monotonic attention assumption of \citet{Cattaneo2020AModel}. This assumption has broad explanatory power: it encompasses many attention-formation processes while avoiding the misspecification concerns associated with stronger parametric assumptions. We show that, under this assumption, a lower bound on the revealed range of $\beta$ can be obtained. The key idea is to exploit violations of regularity induced by external stimuli. If the observed choice data violate the regularity restrictions implied by monotonic internal attention, then, within the model, these violations must be attributed to the influence of external stimuli. This allows us to infer a minimum level of permeation.

This lower bound admits a direct empirical interpretation: it identifies cases in which external stimuli are not merely consistent with the observed choice pattern, but are necessary to explain it under monotonic internal attention. This is particularly valuable in applications where the analyst wants to assess whether an intervention, recommendation, advertisement, or other attentional cue had a meaningful effect on consideration. While the upper bounds in \Cref{Sect:RevealAttInf} discipline how large the effect of external stimuli can be, the lower bound provides evidence that the effect must be positive and quantifies the minimum degree of permeation required by the data.

Finally, in \Cref{Sect:Simulation}, we present simulation results to illustrate the empirical content of the model. We generate choice data from the model under a range of internally monotonic attention processes and examine how the additional monotonicity restriction sharpens inference. The simulations show that this restriction narrows the revealed range of the permeation parameter and produces informative lower bounds that move closely with the true underlying value. We also use the simulations to study the revealed-preference content of the model. The results indicate that the proposed approach can reject many incorrect preference orderings and recover a substantial part of the underlying preference relation, even with limited data.

\section{The framework}\label{Sect:Model}

\subsection{Preliminaries}

Let $X$ be a finite grand set of alternatives, and let $\mathcal{D}$ denote a domain of menus consisting of subsets of $X$. The core of our analysis focuses on the \emph{attention rule}, a central object in the attention-based choice literature (e.g., \citealp{Manzini2014, Cattaneo2020AModel, Brady2016Menu-DependentFeasibility}), which summarizes how frequently different subsets of alternatives are considered. We begin by formally defining this concept.

\begin{definition}[Attention Rule]
    An attention rule $\mu$ is a map $2^X \times \mathcal{D}\rightarrow [0,1]$ such that $\sum_{ A \subseteq S} \mu(A|S)=1$ and $\mu(A|S)=0$ for $A\notin 2^S$ or $A=\emptyset$.  We denote the collection of such attention rule by $\mathcal{M}$. 
\end{definition}

Attention rules have proven to be a powerful tool for revealed preference analysis. To connect attention rules to observable behavior, we next introduce choice data. A (random) choice rule assigns to each alternative $x$ in a menu $S$ the probability with which it is chosen, denoted by $\rho(x|S)$.

\begin{definition}[Choice Rule]
    A choice rule is a map $\rho:X \times \mathcal{D}\rightarrow [0,1]$ such that $\sum_{x \in S} \rho(x|S)=1$ for $S \in \mathcal{D}$, and $\rho(x|S)=0$ for $x \notin S$.  We denote the collection of such choice rule by $\mathcal{C}$.
\end{definition}

Because consideration sets are inherently unobservable, they introduce an additional layer of complexity relative to traditional revealed preference analysis. To conduct revealed preference analysis with limited consideration, it is standard to assume that the observed choice rule is generated by a pair consisting of an (unobserved) attention rule and a preference relation. Specifically, the decision maker (DM) maximizes her preference over a randomly realized consideration set. Randomness in consideration therefore induces randomness in observed choice. For convenience, we define the resulting induced choice rule.

\begin{definition}[Induced Choice Rule]
    Given a strict preference relation $\succ$ and an attention rule $\mu$, the induced choice rule $\mu_\succ \in \mathcal{C}$ is defined by 
 $$\mu_\succ(x|S):=\sum_{A\in S(x,\succ) } \mu(A|S)$$
 where $S(x,\succ) :=\{B\subseteq S: x\succ y\ \forall\ y\in B\setminus\{x\}\}$.
\end{definition}

It is standard in the literature to say that a choice rule $\rho$ is \emph{explained} by an attention rule $\mu$ and a preference relation $\succ$ if $\rho = \mu_{\succ}$. While this formulation offers maximal flexibility in modeling attention, it is well known that the resulting model lacks predictive power. In particular, any observed choice behavior can be rationalized entirely through attention, independently of preferences. It is because one can always match choice probabilities with probabilities over singleton consideration sets, rendering preferences irrelevant. We formalize this observation below.\footnote{Specifically, if we take $\mu$ such that $\mu(\{x\}|S)=\rho(x|S)$ for all $S\in\mathcal{D}$ and all $x\in S$, then $\rho=\mu_\succ$ for any preference $\succ$.}

\begin{remark}\label{rmk:any_choice_go}
 For any $\rho \in \mathcal{C}$ and any $\succ$, there exists $\mu$ such that $\rho=\mu_\succ$.   
\end{remark}

Finally, we introduce a shorthand that will be useful for characterization and identification results.

\begin{definition}[Cumulative Choice Rule] Given a strict preferennce $\succ$, the  upper and lower cumulative choice rules for a choice rule $\rho$  are denoted by $\Urho$ and $\Lrho$ respectively, where 
 $$\Urho(x|S):=1-\sum_{y:   x\succ y } \rho(y|S)  \ \ \ \text{ and } \ \ \ \Lrho(x|S):=1-\sum_{y:   y\succ x } \rho(y|S)  $$
The upper and lower cumulative choice rules for induced choice rule $\mu_\succ$ are analogous.
\end{definition}


\subsection{Attentional Process}

Before introducing the process, we first conceptually discuss the two main objects in our investigation.

\textsc{Internal Attention Allocation---}  
Internal attention allocation captures the decision maker’s endogenous process for directing attention in the absence of external stimuli. This process encompasses goal-directed behavior, habitual patterns, and memory-based mechanisms that shape which alternatives are likely to be considered. For example, a decision maker may systematically attend only to the top-$N$ search results on a webpage, rely on previously successful options recalled from memory, or default to familiar brands due to past experience. In this sense, internal attention allocation reflects pre-existing cognitive structures—such as habits, learned heuristics, salience shaped by past choices, or memory retrieval—that govern attention allocation prior to any external intervention. The outcome of this process is an internal attention rule, generically denoted by $\delta \in \mathcal{M}$, which summarizes the frequency with which different subsets of alternatives are considered when facing a menu. Importantly, the internal attention allocation process is assumed to be inherently unobservable to the analyst, as it is driven by latent cognitive and experiential factors.

\textsc{External Attention Stimuli---}  
External attention stimuli capture any exogenous sources of attention that are presented to the decision maker. The nature of these stimuli depends on the context under study. For instance, when analyzing advertising effects, external stimuli may correspond to the advertisements displayed within a store or on a digital platform. We allow the stimuli to be stochastic so that the same DM can be exposed to different advertisments in different occasions. Such stimuli induce a frequency distribution over \textit{presentation sets}, representing the subsets of alternatives that are externally highlighted to the decision maker. Formally, we denote this distribution by $\gamma \in \mathcal{M}$, where $\gamma(A|S)$ captures how frequently the presentation set $A$ is shown when the available menu is $S$. In contrast to internal attention allocation, we assume that external stimuli—and hence the induced presentation frequencies—are observable.

Having introduced these two constructs, we now describe the two opposing forces that govern attention formation in our model.

\begin{itemize}
    \item \textsc{Proactive Interference}: Pre-existing internal attention inhibits the incorporation of externally presented alternatives. That is, alternatives generated by external stimuli may fail to enter the consideration set because they conflict with, or are crowded out by, internally driven attention.
    
    \item \textsc{Retroactive Interference}: Conditional on external stimuli entering consideration, newly presented alternatives may suppress or displace alternatives that would otherwise be considered under the internal attention rule.
\end{itemize}

In our model, we formalize attention formation as a two-stage process. First, we assume that decision makers exhibit a natural resistance to incorporating new external stimuli into their consideration sets. To capture this resistance, we introduce a \emph{permeation parameter} $\beta \in [0,1]$, which represents the probability that the presentation set $R$ realized from the external stimuli $\gamma$ permeates the process of consideration set formation.

If the external stimuli fail to permeate, the decision maker considers only the set of alternatives $T$ generated by the internal attention rule $\delta$. If, instead, the external stimuli successfully permeate, the set $R$ enters the final consideration set. Conditional on this entry, with probability $\alpha \in [0,1]$—the \emph{displacement parameter}—the decision maker discards the alternatives $T \setminus R$ from consideration, reflecting retroactive interference. With the remaining probability, the alternatives in $T \setminus R$ are retained and jointly considered with $R$.


This process is illustrated in \Cref{fig:AttProcess}. We provide the formal definition below.

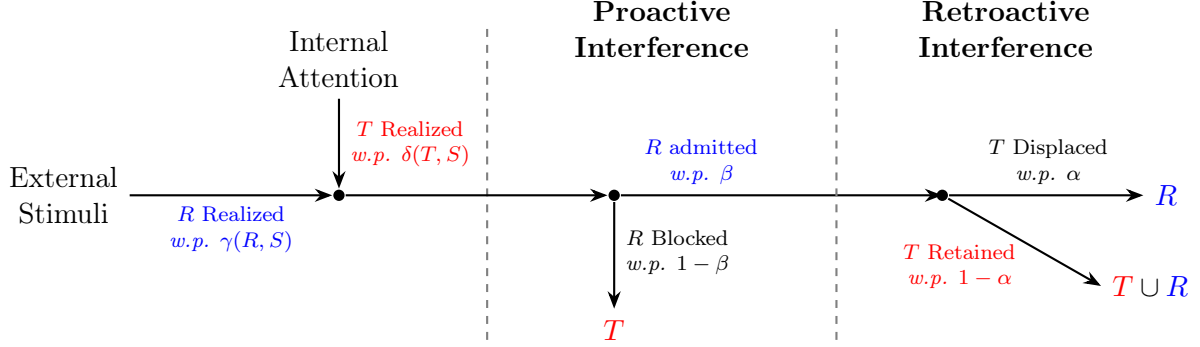
\begin{figure}
    \centering

\begin{tikzpicture}[
    >=Stealth, 
    thick,
    xscale=1.4, 
    dot/.style={circle, fill, inner sep=1.5pt},
    txt/.style={align=center, font=\small}
]

    \node (gamma) at (-0.8,0) [txt] {External\\Stimuli};
    
    \node[dot] (j1) at (1.8,0) {};
    
    \draw[->] (gamma) -- node[below, align=center, font=\scriptsize, text=blue] {$R$ Realized \\ \textit{w.p.} $\gamma(R,S)$} (j1);
    
    \node (delta) at (1.8,1.8) [txt] {Internal\\Attention};
    
    \draw[->] (delta) -- node[right, align=center, font=\scriptsize, text=red] {$T$ Realized \\ \textit{w.p.} $\delta(T,S)$} (j1);

    \draw[dashed, gray] (3.2, 2.2) -- (3.2, -2);
    
    \node at (4.85, 2.2) [txt] {\textbf{Proactive} \\ \textbf{Interference}};

    \node[dot] (j2) at (4.4,0) {};
    \draw[->] (j1) -- (j2);
    
    \draw[->] (j2) -- (4.4,-1.5) node[below, txt] {${\color{red}T}$};
    \node[right, align=center, font=\scriptsize] at (4.4,-0.75) {$R$ Blocked \\ \textit{w.p.} $1-\beta$};
    
    \draw[->] (j2) -- (7.5,0) 
        node[pos=0.25, above, align=center, font=\scriptsize, text=blue] {$R$ admitted \\ \textit{w.p.} $\beta$};

    \draw[dashed, gray] (6.5, 2.2) -- (6.5, -2);
    
    \node at (8.1, 2.2) [txt] {\textbf{Retroactive} \\ \textbf{Interference}};

    \node[dot] (j3) at (7.5,0) {};
    
    
    \draw[->] (j3) -- (9.4,0) node[right, txt] {${\color{blue}R}$};
    \node[above, align=center, font=\scriptsize] at (8.5,0) {$T$ Displaced \\ \textit{w.p.} $\alpha$};
    
    \draw[->] (j3) -- (9,-1.2) node[right, txt] {${\color{red}T} \cup {\color{blue}R}$};
    \node[below left, align=center, font=\scriptsize, text=red] at (8.3,-0.5) {$T$ Retained \\ \textit{w.p.} $1-\alpha$};

\end{tikzpicture}

    \caption{Attentional Process}
    \label{fig:AttProcess}
\end{figure}

\begin{definition}[$\alpha$-$\beta$ Interference]
    We say that an attention rule $\mu \in \mathcal{M}$ results from $\alpha$-$\beta$ interference exerted by external stimuli $\gamma$  on an internal attention rule $\delta$ if there exists an internal allocation rule $\delta \in \mathcal{M}$, permeability parameter $\beta \in [0,1]$ and displacement parameter $\alpha \in [0,1]$ such that 
$$\mu(A|S)=(1-\beta)\delta(A|S) + \beta \left( \alpha \gamma(A|S)  + (1-\alpha)\sum_{(B,C): B\cup C =A} \gamma(B|S)\delta(C|S) \right).$$ 
\end{definition}



Following the existing literature, we assume that the decision maker (DM) is endowed with a strict preference relation (i.e., a linear order) over the set of available alternatives. We now provide the formal definition of the model.


\begin{definition}[Attentional Interference Model] \label{Definition: AIM}
A pair of choice rule and external stimuli  $(\rho,\gamma)$ is  represented by  Attentional Interference model (AIM)  if there exists a linear order $\succ$, an attention rule $\mu$  resulting from $\alpha$-$\beta$ interference exerted by $\gamma$ on $\delta$ such that $\mu_\succ = \rho$.  In such a case, we say $(\rho,\gamma)$ admits an $(\succ,\mu, \alpha,\beta,\delta)$ AIM representation.
\end{definition}

\section{Revealed Attentional Interference} \label{Sect:RevealAttInf}
Guided by the Attention Interference Model (AIM), our objective is to infer the extent to which individuals’ attention is affected by external stimuli. Our analysis proceeds in two levels of specification. First, we study the general version of the model and provide a characterization of the set of admissible parameters. We then examine three special cases of the model in which the displacement parameter is fixed at $1$, $0$, or equal to the permeability parameter $\beta$. In the subsequent section, we further restrict the model by imposing an additional assumption on the internal attention rule. Through this progression, we aim to demonstrate how increasingly strong, nested assumptions on the attention process sharpen the range of feasible inference. 

We begin by formally defining the inference problem. Following the conventional conservative approach in the revealed attention literature (e.g., \citealp{Masatlioglu2012RevealedAttention}), a value of $\beta$ is rejected if there exists no model representation consistent with the observed data.

\begin{definition}[Revealed Attentional Interference]\label{Def:RevealedAtteInf}
We define the identified range of attentional interference as
$$\BI(\rho,\gamma):=\{ (\alpha,\beta) \in \mathcal{A} : \text{$(\rho,\gamma)$ admits an $(\succ,\mu, \alpha,\beta,\delta)$ AIM representation for some} \succ, \mu \text{ and } \delta\}.$$
where $\mathcal{A}:=[0,1]\times[0,1)$.
\end{definition}

The identified range is inherently model-based. Consequently, establishing its non-emptiness requires the existence of an underlying representation consistent with the data. This motivates an axiomatic analysis, which provides necessary and sufficient conditions for such a representation to hold given a dataset $\mathcal{D}$. For expositional clarity, we exclude the boundary case $\beta = 1$. This restriction is innocuous and can be relaxed at the cost of introducing additional boundary considerations without affecting the substance of our results.\footnote{Therefore, we also implicitly assume away the knife-edge case that $\rho=\gamma_\succ$ for some $\succ$.}

\subsection{Identification}\label{Sect: general}

We begin with a simple but important observation. The following proposition states that the identified range always includes zero permeability of external stimuli, regardless of the degree of displacement.

\begin{prop}[Existence of Zero-Permeability Explanation]\label{Prop: Close interval of beta}
    $(\alpha,0) \in \BI(\rho,\gamma)$ for any $\alpha \in [0,1]$ 
\end{prop}

Given a choice data $\rho$ and external stimuli $\gamma$, one cannot rule out the possibility that the stimuli always never permeate through consideration set. This is based on the well-known fact in the attention literature where one can attribute everything to attention when explaining choice probability by employing singleton consideration set as stated in \Cref{rmk:any_choice_go}. Therefore, an immediately corollary is that any choice probability can be explained under this model with any preference $\succ$ and any displacement parameter $\alpha$.

\begin{coro}\label{Coro:AnyRhoAIM}
For every $\succ$ and every $\alpha$, any $(\rho,\gamma)$ admits  an $(\succ,\mu,\alpha,0,\delta)$ AIM representation for some $\mu$ and $\delta$.
\end{coro}


This observation naturally raises the question of what can be learned about the impact of external stimuli on observed choice behavior. Although zero permeability can never be ruled out, this does not preclude policymakers or researchers from claiming that external stimuli have a substantial effect. The more relevant question is therefore the following: to what extent can we reject theoretically infeasible claims of a high degree of external influence?

To address this question, we seek to identify an upper bound on the impact of external stimuli. We begin by deriving a weak bound on the admissible values of the displacement and permeability parameters, $\alpha$ and $\beta$. Throughout the analysis, we adopt the convention that $\frac{a}{0} = \infty$ for any $a \geq 0$.

\begin{prop}[Weak Bound]\label{prop:WeakBound}
Let $(\rho,\gamma)$ admit an $(\mu,\succ,\alpha,\beta,\delta)$ AIM representation.
Let $w_S:=\min(\succ,S)$ and $b_S:=\max(\succ,S)$. Then, 
\begin{enumerate}
\item[(i)] \emph{(Bounded by the worse)}
\begin{equation}\tag{WB-w}
\alpha\beta \le 
\frac{\rho(w_S|S)}{\gamma_\succ(w_S|S)} \ \ \text{for all \ } S\in\mathcal{D} 
\end{equation}

\item[(ii)] \emph{(Bounded by the best)}
\begin{equation}\tag{WB-b}
\beta \le 
\frac{\rho(b_S|S)}{\gamma_\succ(b_S|S)} \ \  \text{for all \ } S\in\mathcal{D} 
\end{equation}
\end{enumerate}
\end{prop}

\begin{proof}
    For statement (i), note that the worst alternative is chosen if and only if it is the only considered option. Note also that $\mu(A|S)\geq\alpha\beta\gamma(A|S)$ for all $A,S$. Thus,
    $$\rho(w_S|S)=\mu(\{w_S\}|S)\geq\alpha\beta\gamma(\{w_S\}|S)=\alpha\beta\gamma_\succ(w_S|S).$$
    This shows (i).

    For statement (ii), note that the best alternative is chosen whenever it is considered. Note also that
    \begin{align*}
    \sum_{A\subseteq S:b_S\in A}\mu(A|S)&\geq\beta \left( \alpha \sum_{A\subseteq S:b_S\in A}\gamma(A|S)  + (1-\alpha)\sum_{A\subseteq S:b_S\in A}\sum_{(B,C): B\cup C =A} \gamma(B|S)\delta(C|S) \right)\\
    &\geq\beta\sum_{A\subseteq S:b_S\in A}\gamma(A|S).
    \end{align*}
    Hence we have
    $$\rho(b_S|S)=\sum_{A\subseteq S:b_S\in A}\mu(A|S)\geq\beta\sum_{A\subseteq S:b_S\in A}\gamma(A|S)=\beta\gamma_\succ(b_S|S).$$
    So (ii) holds.
\end{proof}

Notice that \Cref{prop:WeakBound} provides two simple and intuitive upper bounds on the influence parameters. We first focus on part (i), which makes use of the worst alternative according to the underlying preference in the representation. To understand this bound, consider an extreme example in which $\gamma_{\succ}(w_S|S)=1$ but $\rho(w_S|S)=0.1$. By definition, $\gamma_{\succ}(w_S|S)=\gamma(\{w_S\}|S)$, implying that the worst alternative is always presented in isolation, yet it is chosen only $10\%$ of the time. 

First, consider the case of perfect displacement, $\alpha=1$. The bound then implies that $\beta$ cannot exceed $0.1$. Under perfect displacement, the alternative $w_S$ is chosen whenever it permeates the consideration set. Hence, it can permeate attention at most $10\%$ of the time. Second, consider the case of perfect permeation, $\beta=1$. In this case, the bound implies that $\alpha$ cannot exceed $0.1$. Under perfect permeation, the alternative $w_S$ is chosen whenever it displaces the initial consideration set. Thus, it can displace internal attention at most $10\%$ of the time.

Part (ii) of the proposition provides a complementary bound that depends only on the permeation parameter $\beta$. The interpretation is analogous. Consider the same example in which $\gamma_{\succ}(b_S|S)=1$ but $\rho(b_S|S)=0.1$. By definition,
$\gamma_{\succ}(b_S|S)=\sum_{A \ni b_S}\gamma(A|S)$,
so every presentation set includes the best alternative $b_S$ according to the underlying preference. Since $b_S$ is the best alternative, it is chosen whenever it permeates the consideration set, regardless of whether it displaces the initial consideration set. Consequently, the external stimuli can permeate attention at most $10\%$ of the time.

Moreover, to obtain a complete characterization of $\BI(\rho,\gamma)$, the identified set must be linked to a representation that specifies an underlying preference relation. To this end, we introduce the \emph{preference-specific} revealed range of attention parameters:
\[
\BI_{\succ}(\rho,\gamma)
:= \{ (\alpha,\beta) \in \mathcal{A} :
(\rho,\gamma) \text{ admits an } (\succ,\mu,\alpha,\beta,\delta)\text{-AIM representation for some }
\mu \text{ and } \delta \}.
\]
By definition, it follows immediately that
\[
\BI(\rho,\gamma) = \bigcup_{\succ \in \mathcal{P}} \BI_{\succ}(\rho,\gamma).
\]

Before stating the main theorem, we introduce a key inequality—referred to as the \emph{Attentional Interference Inequality}—which will play a central role in the characterization that follows.

\begin{definition}
The Attentional Interference Inequality is a point-wise inequality  over functions $\rho,\Urho,\gammaS,\Ugamma$ with parameters $\alpha,\beta$ such that
\begin{equation}\tag{AII}\label{eq:RAI}
\alpha\beta \Big(\rho(1-\Ugamma)+\gammaS(\Urho-\beta)\Big)
\;\le\;
\rho(1-\beta\Ugamma)+\beta\gammaS(\Urho-1).
\end{equation}
\end{definition}

With this definition, we are really to state the main result for the general model.

\begin{thm}[Revealed Attentional Interference]\label{Thm:RevealAttInt}
Given a dataset $\mathcal{D}$, $$\BI_\succ(\rho,\gamma)
=
\Bigl\{(\alpha,\beta)\in \mathcal{A}: (\alpha,\beta)\text{ satisfies (AII)}\Bigr\}.$$ Furthermore, if $S\in\mathcal{D}$ implies $|S|=2$, then
$$\BI_\succ(\rho,\gamma)
=
\Bigl\{(\alpha,\beta)\in \mathcal{A}:
(\alpha,\beta) \text{ satisfies (WB-b) and (WB-w)}
\Bigr\}.$$
\end{thm}

\begin{proof}
    The proof is in \Cref{App:Proofs}.\footnote{All proofs are in the appendix unless specified otherwise.}
\end{proof}

\Cref{Thm:RevealAttInt} provides a complete characterization of the range of values that the parameters $\alpha$ and $\beta$ can take, given the observed choice data $\rho$ and external stimuli $\gamma$. Owing to the nonlinear nature of the attentional process, the Attentional Interference Inequality (AII) is inherently nonlinear in the parameters; in particular, it is quadratic in $\beta$. Interestingly, the second part of the theorem shows that when the dataset $\mathcal{D}$ consists only of menus with two alternatives, the weak bounds derived in \Cref{prop:WeakBound} are tight. In this case, the Attentional Interference Inequality reduces to the bounds (WB-b) and (WB-w) established in \Cref{prop:WeakBound}.

To see this, consider $S$ such that $|S|=2$. Hence, it must be that $S=\{b_S,w_S\}$ for a given $\succ$. We first consider the case of the worse alternative $w_S$ in $S$. Since $w_S$ is worse in $S$, $\Ugamma(w_S|S)=\Urho(w_S|S)=1$.  We skip denoting $(w_S|S)$ for clarity. Hence, (AII) reduces to 
\begin{align*}
\alpha \beta \Big(\rho(1-1)+\gammaS(\rho-\beta) \Big)&\leq \rho(1-\beta)+\beta\gammaS(1-1) \\
    \alpha \beta \gammaS(1-\beta)& \leq \rho(1-\beta)\\
     \alpha \beta &\le \frac{\rho}{\gammaS} 
\end{align*} where $\beta<1$. Hence, it gives the (WB-w) for $w_S \in S$. On the other hand, consider the case of the best alternative $b_S$ in $S$. Since $b_S$ is the best in $S$, we have $\Ugamma(b_S|S)=\gammaS(b_S|S)$ and $\Urho(b_S|S)=\rho(b_S|S)$. We skip denoting $(b_S|S)$ for clarity.  Hence, (AII) reduces to 
\begin{align*}
\alpha \beta \Big(\rho -\gammaS\rho+\gammaS\rho-\beta\gammaS\Big)& \leq \rho -\rho\beta\gammaS +\rho\beta\gammaS-\beta\gammaS\\
    \alpha \beta ( \rho -\beta\gammaS) &\leq  \rho  -\beta\gammaS \\
   0&\leq (1-\alpha\beta)(\rho-\beta\gammaS)\\ 
   \beta &\leq \frac{\rho}{\gammaS}
\end{align*} where $\alpha\beta-1<0$. Hence, it gives (WB-b) for $b_S \in S$. 

More generally, the shape and bounds of the Attentional Interference Inequality (AII) depend on the specific choice data $\rho$ and external stimuli $\gamma$. To illustrate this, \Cref{fig:AIIExample} plots the AII for three distinct cases by varying $\Urho$ while keeping other values fixed at $(\rho, \gammaS, \Ugamma) = (0.15, 0.3, 0.8)$. 

In all three cases, the decision maker chooses $x$ less frequently than the external stimuli suggest ($\rho = 0.15 < 0.3 = \gammaS$), and $x$ is a relatively low-ranked alternative within the stimuli ($\gammaS < \Ugamma = 0.8$). A common feature across all three panels is that in the extreme case of Perfect Displacement ($\alpha=1$), the permeation parameter $\beta$ is bounded above by $0.5$ (indicated by the dashed gray line). This threshold implies that under full displacement, since the choice only matches the external stimuli half of the time, permeation can only be effective half of the time.

The transition across the panels highlights how the AII for a specific $(x,S)$ behaves as we approach the opposite extreme of Perfect Retention ($\alpha=0$). In the left panel ($\Urho=0.8$), the choice data suggests $x$ is highly retained, allowing for a wider range of $\beta$ even at low displacement levels. In the middle panel ($\Urho=0.4$), the boundary simplifies to a constant horizontal line, making the permeation limit independent of $\alpha$. Finally, in the right panel ($\Urho=0.15$), where $\Urho$ is low, the bound becomes significantly more restrictive as $\alpha$ decreases, ``pinching'' the allowed values of $\beta$ toward the origin. 

While this may suggest that the thresholds of permeation at either extreme of displacement do not dominate one another, this is not the case when considering the AII over the entire domain of the function. In fact, one can make more specific predictions by focusing on the revelation of permeation, which we discuss in greater detail in the next section.

\begin{figure}[ht]
\centering

\begin{tikzpicture}
\begin{axis}[
  title={(0.15, 0.8, 0.3, 0.8)},
  xlabel={$\alpha$},
  ylabel={$\beta$},
  xmin=0, xmax=1,
  ymin=0, ymax=1,
  xtick={0, 1},
  ytick={0, 0.5, 1},
  axis lines=left,
  width=5cm, height=4.5cm,
  grid=both,
  samples=400,
  ylabel style={rotate=-90}, 
]
\pgfmathsetmacro{\rho}{0.15}
\pgfmathsetmacro{\Urho}{0.8}
\pgfmathsetmacro{\gammaS}{0.3}
\pgfmathsetmacro{\Ugamma}{0.8}
\pgfmathdeclarefunction{abound}{1}{\pgfmathparse{(\rho*(1-#1*\Ugamma) + #1*\gammaS*(\Urho-1)) / (#1*(\rho*(1-\Ugamma) + \gammaS*(\Urho-#1)))}}

\draw[gray, dashed, thin] (axis cs:0,0.5) -- (axis cs:1,0.5);

\path[fill=red, fill opacity=0.25]
  (axis cs:0,0) -- (axis cs:1,0) -- (axis cs:1,0.15)
  plot[variable=\b, domain=0.15:0.85] ({abound(\b)},{\b})
  -- (axis cs:0,0) -- cycle;

\addplot[red, very thick, domain=0.15:0.85, variable=\b] ({abound(\b)},{\b});
\end{axis}
\end{tikzpicture}
\hfill
\begin{tikzpicture}
\begin{axis}[
  title={(0.15, 0.4, 0.3, 0.8)},
  xlabel={$\alpha$},
  ylabel={$\beta$},
  xmin=0, xmax=1,
  ymin=0, ymax=1,
  xtick={0, 1},
  ytick={0, 0.5, 1},
  axis lines=left,
  width=5cm, height=4.5cm,
  grid=both,
  samples=400,
  ylabel style={rotate=-90}, 
]
\draw[gray, dashed, thin] (axis cs:0,0.5) -- (axis cs:1,0.5);

\path[fill=green!60, fill opacity=0.3]
  (axis cs:0,0) -- (axis cs:1,0) -- (axis cs:1,0.5) -- (axis cs:0,0.5) -- cycle;

\addplot[green!60!black, very thick] coordinates {(0,0.5) (1,0.5)};
\end{axis}
\end{tikzpicture}
\hfill
\begin{tikzpicture}
\begin{axis}[
  title={(0.15, 0.15, 0.3, 0.8)},
  xlabel={$\alpha$},
  ylabel={$\beta$},
  xmin=0, xmax=1,
  ymin=0, ymax=1,
  xtick={0, 1},
  ytick={0, 0.5, 1},
  axis lines=left,
  width=5cm, height=4.5cm,
  grid=both,
  samples=400,
  ylabel style={rotate=-90}, 
]
\pgfmathsetmacro{\rho}{0.15}
\pgfmathsetmacro{\Urho}{0.15}
\pgfmathsetmacro{\gammaS}{0.3}
\pgfmathsetmacro{\Ugamma}{0.8}

\pgfmathdeclarefunction{ybound}{1}{%
  \pgfmathparse{(%
    -(#1*(\rho*(1-\Ugamma)+\gammaS*\Urho) + \rho*\Ugamma - \gammaS*(\Urho-1))%
    + sqrt(max(0, (#1*(\rho*(1-\Ugamma)+\gammaS*\Urho) + \rho*\Ugamma - \gammaS*(\Urho-1))^2 - 4*(-#1*\gammaS)*(-\rho)))%
  ) / (2*(-#1*\gammaS))}%
}

\draw[gray, dashed, thin] (axis cs:0,0.5) -- (axis cs:1,0.5);

\path[fill=blue!60, fill opacity=0.3]
  (axis cs:0,0) -- (axis cs:0,0.4) 
  plot[variable=\x, domain=0.001:1] ({\x},{ybound(\x)})
  -- (axis cs:1,0) -- (axis cs:0,0) -- cycle;

\addplot[blue!60!black, thick, domain=0.001:1, variable=\x] ({\x}, {ybound(\x)});
\draw[blue!60!black, thick] (axis cs:0,0) -- (axis cs:0,0.4);

\end{axis}
\end{tikzpicture}

\caption{Comparison of AII for different $(\rho, \Urho, \gammaS, \Ugamma)$ with a reference line at $\beta=0.5$. Left: (0.15, 0.8, 0.3, 0.8); Middle: (0.15, 0.4, 0.3, 0.8); Right: (0.15, 0.15, 0.3, 0.8).}
\label{fig:AIIExample}
\end{figure}
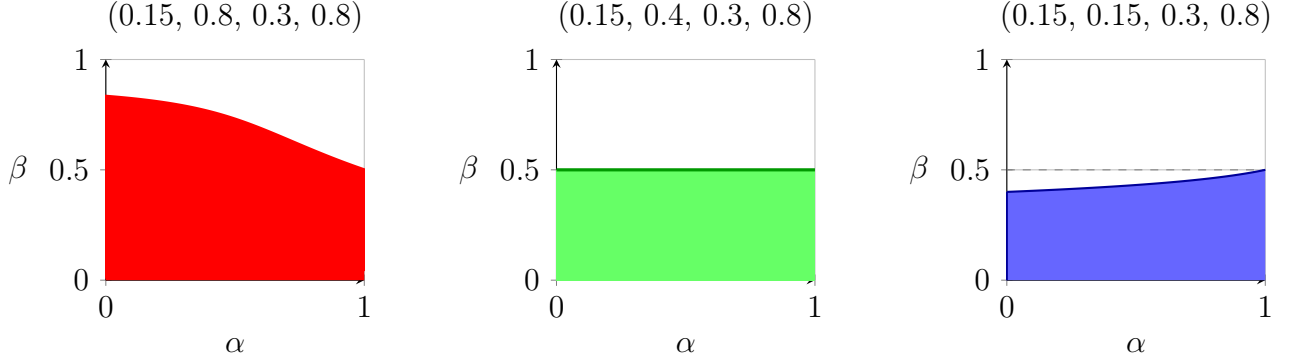

Finally, we provide a sketch of the proof. The argument relies on the observation that one can construct the underlying cumulative induced choice rule $\Udelta$ using the observable objects $(\rho,\gamma)$ together with the parameters $\alpha$ and $\beta$. By \Cref{rmk:any_choice_go}, as long as we can ensure that the induced choice rule is a well-defined cumulative distribution, it can be supported by some underlying attention rule $\delta$. The requirement that the cumulative distribution be increasing can be expressed as a system of inequalities involving $\rho$, $\Urho$, $\gamma_{\succ}$, and $\Ugamma$. Solving these inequalities and rearranging terms yields the Attentional Interference Inequality (AII).




\subsection{Identifying Permeation Under Special Cases}\label{Sect:3Cases}

In this section, we focus on identifying the key parameter of interest in our setting—the permeation parameter. Permeation captures whether external stimuli succeed in attracting attention, which we view as a first-order metric of interest to policymakers and practitioners alike. Whether an advertisement or promotional campaign is effective ultimately hinges on its ability to induce consideration of the options it seeks to promote. To this end, we study three benchmark models. In what follows, we allow complete freedom in the value of the permeation parameter $\beta$ and concentrate on bounding its admissible range. Accordingly, we impose restrictions on the displacement parameter $\alpha$. We consider three cases: \emph{Perfect Displacement}, where $\alpha=1$; \emph{Perfect Retention}, where $\alpha=0$; and \emph{Unified Interference}, where $\alpha=\beta$.


We now formally define the three benchmark models.

\begin{definition}[Three Special Cases] \label{Definition: AIM}
A pair $(\rho,\gamma)$ is  represented by  Attentional Interference model with Perfect Displacement (AIM-D)   if it has a $(\succ,\mu, 1,\beta,\delta)$ AIM representation. A  $(\rho,\gamma)$ is  represented by  Attentional Interference model with Perfect Retention (AIM-R)   if it has a $(\succ,\mu, 0,\beta,\delta)$ AIM representation. A  $(\rho,\gamma)$ is  represented by  Unified Attentional Interference model (UAIM)  if it has a $(\succ,\mu, \beta,\beta,\delta)$ AIM representation. 
\end{definition}

Again, by \Cref{Coro:AnyRhoAIM}, we know that any choice rule $\rho$ admits all three representations. Nevertheless, our interest lies in identifying the set of parameter values that are admissible under each special case. For notational simplicity, we omit explicit definitions of the corresponding revealed ranges. However, each analysis effectively characterizes the full revealed range of the permeation parameter $\beta$, as the results are all derived from \Cref{Thm:RevealAttInt}.\footnote{Effectively, for $i=D,R,U$, we have \[
 \{ \beta \in [0,1) :
(\rho,\gamma) \text{ admits an } \text{AIM-} i \text{ representation for} \succ \text{and some } \mu \text{ and } \delta \}=[0, \min_{(x,S)} \beta_\succ^i(x,S)  ]. \]}




\begin{prop}[Revealed Bound Under Special Cases]\label{prop:BoundSpecial}
Let $(\rho,\gamma)$ admit an $(\mu,\succ,\alpha,\beta,\delta)$ AIM representation. Then, 
\begin{enumerate}
\item[(i)] \emph{(Perfect Displacement)} if $\alpha=1$, for all   $x\in S\in\mathcal{D}$
\begin{equation*}\tag{AIM-D-UB}
\beta \leq \frac{\rho(x|S)}{\gammaS(x|S)} =:\beta^{D}_\succ(x,S)  
\end{equation*}

\item[(ii)] \emph{(Perfect Retention)} if $\alpha=0$, for all   $x\in S\in\mathcal{D}$
\begin{equation*}\tag{AIM-R-UB}
\beta \leq \frac{\rho(x|S)}{\gammaS(x|S)(1-\Urho(x|S))+\rho(x|S) \Ugamma(x|S)}=:\beta^{R}_\succ(x,S)
\end{equation*}

\item[(iii)] \emph{(Unified Interference)} if $\alpha=\beta$, for all   $x\in S\in\mathcal{D}$
\begin{equation*}\tag{UAIM-UB}
\beta \leq  \frac{\Phi(x,S)+\sqrt{\Phi(x,S)^2+4\rho(x|S)\gammaS(x|S)}}{2\gammaS(x|S)} =:\beta^{U}_\succ(x,S)
\end{equation*}
where $\Phi(x,S):= \rho(x|S)(1-\Ugamma(x|S))-\gammaS(x|S)(1-\Urho(x|S))$.
\end{enumerate}
\end{prop}


We first consider the case of Unified Attentional Interference. In this case, the bound involves a square root. Since the Attentional Interference Inequality is a polynomial inequality of degree three, imposing $\alpha=\beta$ yields a cubic inequality in $\beta$. Nevertheless, two roots can be ruled out, leaving a unique admissible root, as shown in part (iii). To visualize the bound, one may draw a $45^\circ$ line in the plots in \Cref{fig:AIIExample}; the intersection of this line with the relevant AII boundary identifies the solution. 

Then, we compare the cases of Perfect Displacement and Perfect Retention. Observe that there is no fixed ordering between the bounds for a given $(x,S)$ across these two cases. In particular, one can show that the bound under AIM-R exceeds that under AIM-D if and only if
\begin{align*}
  \beta^{D}_\succ \leq \beta^{R}_\succ\iff  \frac{\rho}{\gammaS} \leq \frac{\rho}{\gammaS(1-\Urho)+\rho \Ugamma}
    \iff
    \frac{\Ugamma}{\gammaS} \leq \frac{\Urho}{\rho}.
\end{align*}
This condition can be verified using the examples in \Cref{fig:AIIExample}. In the cases $(\rho,\Urho,\gammaS,\Ugamma)=(0.15,0.8,0.3,0.8)$ at a given $(x,S)$, we have $\frac{\Ugamma}{\gammaS} < \frac{\Urho}{\rho}$, and in the case of  $(0.15,0.15,0.3,0.8)$, we have $\frac{\Ugamma}{\gammaS} > \frac{\Urho}{\rho}$. These relationships correspond to the two edge points observed in the plots. However, across the entire domain, we will see that the inequality holds only on one side. To demonstrate this, we first introduce a general result regarding the revealed set.

\begin{thm}[Monotonicity of Displacement]\label{Thm:Monotonicity_alpha} For any dataset $\mathcal{D}$,     $(\alpha,\beta) \in \BI_\succ(\rho,\gamma)$ implies  $ (\alpha',\beta) \in \BI_\succ(\rho,\gamma)$ for $\alpha' < \alpha$.
\end{thm}

This result indicates that whenever the choice data can be explained by the model with a specific $(\alpha, \beta)$, it can also be explained using a smaller $\alpha$. Consequently, a smaller $\alpha$ is less conservative when bounding the effect of permeation. Intuitively, a smaller $\alpha$ provides the model more flexibility to explain behavior through internally generated attention, as these internal considerations are not displaced. In the spirit of \Cref{rmk:any_choice_go}, this internally generated consideration allows for significant freedom in explaining choice data. This flexibility, in turn, allows for a greater claimed effect of permeation. Notably, this result immediately suggests that the bound for permeation under Perfect Retention is always greater than or equal to that under Perfect Displacement, contrary to the pointwise intuition suggested by \Cref{fig:AIIExample}. We state this result in the following result.\footnote{While this result is derived as a corollary here, we also provide a direct proof of this result in the appendix.}


\begin{prop}[Ordering Permeation Bound] 
\label{coro:OrderingBound} For any dataset $\mathcal{D}$,
     $$\min_{(x, S)}\beta^{D}_\succ(x,S)\leq\min_{(x, S)}\beta^{U}_\succ(x,S) 
     \leq \min_{(x, S)}\beta^{R}_\succ(x,S)$$
\end{prop}
\begin{proof}
    Observe that 
    \begin{align*}
    \min_{(x, S)}\beta^{D}_\succ(x,S)&=\max\{\beta:(1,\beta)\in\BI_\succ(\rho,\gamma)\};\\
    \min_{(x, S)}\beta^{R}_\succ(x,S)&=\max\{\beta:(0,\beta)\in\BI_\succ(\rho,\gamma)\};\\
    \min_{(x, S)}\beta^{U}_\succ(x,S)&=\max\{\beta:(\beta,\beta)\in\BI_\succ(\rho,\gamma)\}.
    \end{align*}
    By \Cref{Thm:Monotonicity_alpha}, if $(\beta,\beta)\in\BI_\succ(\rho,\gamma)$ then  $(0,\beta)\in\BI_\succ(\rho,\gamma)$. So $\min_{(x, S)}\beta^{U}_\succ(x,S)\in\{\beta:(0,\beta)\in\BI_\succ(\rho,\gamma)\}$. Hence, $\min_{(x,S)}\beta^{U}_\succ(x,S)\leq\min_{(x,S)}\beta^{R}_\succ(x,S)$. 
    
    Similarly, if $(1,\beta)\in\BI_\succ(\rho,\gamma)$ then  $(\beta,\beta)\in\BI_\succ(\rho,\gamma)$ by \Cref{Thm:Monotonicity_alpha}. So $\min_{(x, S)}\beta^{D}_\succ(x,S)\in\{\beta:(\beta,\beta)\in\BI_\succ(\rho,\gamma)\}$. Hence, $\min_{(x,S)}\beta^{D}_\succ(x,S)\leq\min_{(x,S)}\beta^{U}_\succ(x,S)$.
\end{proof}

This result also suggests that the allowable value for permeation under Perfect Displacement is more ``conservative'' than under Perfect Retention. In this corollary, it also shows that the case of unified interference sits between the two extreme cases. In the next section, we use Perfect Displacement as an illustration for incorporating additional assumptions regarding internal attention. As a final mark, \Cref{Thm:Monotonicity_alpha} also indicates that the global maximum permeation parameter is given by the case where $\alpha=0$, which we state in the following corollary.

\begin{coro}[Global Maximum] 
\label{coro:max_beta}The global maximal $\beta$ is identified at $\alpha=0$ as
$$\max\{\beta:(\alpha,\beta)\in\BI(\rho,\gamma)\ \text{for some}\ \alpha\}=\max_\succ\min_{(x, S)}\beta^{R}_\succ(x,S)$$
\end{coro}

\section{Non-parametric Restriction over Internal Attention} \label{Sect:NonPara}

In the previous section, we analyzed the general case in which the internal attention rule $\delta$ is allowed to be completely flexible. While this flexibility enables the model to accommodate any form of internal attention formation, it also prevents us from ruling out the possibility of zero influence: the permeation parameter $\beta$ can always be set equal to zero to rationalize the observed choice data together with external stimuli. Nevertheless, imposing additional structure on the internal attention rule allows the model to deliver sharper implications. In particular, such restrictions can generate lower bounds on the parameters and, at the same time, endow the model with falsifiability and revealed preference content.

In this section, we illustrate this idea by adopting a prominent assumption from the attention literature—\emph{monotonic attention}—introduced by \citet{Cattaneo2020AModel}. Monotonic attention requires that the probability of a given subset being considered cannot decrease when an alternative outside that subset becomes infeasible. Many commonly used attention rules satisfy this property, including the independent consideration rule of \citet{Manzini2014}, the logit attention rule of \citet{Brady2016Menu-DependentFeasibility}, and the random categorization model of \citet{Aguiar_2017_EL}. As such, monotonic attention constitutes a natural and empirically relevant restriction on the internal attention rule.

We now formalize this concept. An internal attention rule $\delta$ is said to be \emph{monotonic} if, for every $S \in \mathcal{D}$, every $A \subseteq S$, and every $x \in S \setminus A$, we have
\[
\delta(A|S \setminus \{x\}) \geq \delta(A|S).
\]

With this definition in place, we proceed to study our model of interest. In what follows, we focus on the case of Perfect Displacement. We adopt this version of the model for three reasons. First, it provides the strongest bound over permeation. Second, it is conservative in the sense that it only allows a policy maker to claim a weaker influence over permeation. Third, it admits a particularly transparent characterization. Nonetheless, the analysis can, in principle, be extended to other variants of the model.

\begin{definition}[AIM-DM] \label{Definition: AIM}
A pair $(\rho,\gamma)$ is  represented by  Attentional Interference model with Perfect Displacement and Monotonicity (AIM-DM)   if it has a $(\succ,\mu, \beta,\delta)$ AIM-D representation and $\delta$ is monotonic. In this case, we say that $(\rho,\gamma)$ admits a $(\succ,\mu,\beta,\delta)$ AIM-DM representation. 
\end{definition}

We then define our object of interest, which preserves the essence of the model-based inference approach introduced in \Cref{Def:RevealedAtteInf}. As before, establishing the non-emptiness of the identified range requires the existence of an underlying representation consistent with the observed data.

\begin{definition}[Revealed Attentional Interference under AIM-DM]
We define the identified range of attentional interference as
\begin{align*}
    \BB^{DM}(\rho,\gamma):=\{ \beta \in [0,1) : &(\rho,\gamma) \text{ admits an $(\succ,\mu,\beta,\delta)$ AIM-DM representation}\\ 
    &\text{for some }  \succ, \mu \text{ and }  \delta \}.
\end{align*}
\end{definition}

When the internal attention rule $\delta$ is required to be monotonic, certain attention rules can no longer be used to support a given value of $\beta$ in an AIM-DM representation. As a result, for fixed choice data $\rho$ and external stimuli $\gamma$, the admissible range of the permeation parameter $\beta$ becomes smaller than in the unrestricted case. As before, we focus on preference-specific identification and work with a fixed preference relation $\succ$. Accordingly, we define
\begin{align*}
    \BB^{DM}_\succ(\rho,\gamma):=\{ \beta \in [0,1) : &(\rho,\gamma) \text{ admits an $(\succ,\mu,\beta,\delta)$ AIM-DM representation}\\ 
    &\text{for some }  \mu \text{ and }  \delta \}.
\end{align*}

It follows immediately that
\[
\BB^{DM}(\rho,\gamma) = \bigcup_{\succ \in \mathcal{P}} \BB^{DM}_{\succ}(\rho,\gamma).
\]


\subsection{Identification and Characterization}

For notational convenience, we define, for any choice rule $\rho \in \mathcal{C}$, a \textit{differencing} operator $\Delta_x$ acting on $\rho(y|S)$, which captures the change in choice probability when an alternative is removed from the menu. Formally, for $x \in S$, define
\[
\Delta_x \rho(y|S) := \rho(y|S \setminus \{x\}) - \rho(y|S).
\]

In the literature, it is conventional to say that a choice rule $\rho$ is \emph{regular} if $\Delta_x \rho(y|S)$ is always nonnegative; that is, an alternative becomes (weakly) more likely to be chosen when another alternative is removed from the menu. Regularity is closely related to the characterization of monotonic attention.

\citet{Cattaneo2020AModel} demonstrate that the implications of monotonic attention can be expressed as restrictions on how the choice probability of an alternative $y$ may change when another alternative $x$ becomes infeasible. We now state a version of their characterization theorem for monotonic attention, adapted to our setting.

\begin{thm}[\citet{Cattaneo2020AModel}]\label{Thm:RAM} For $\mathcal{D}=2^X \setminus \emptyset $, there exists monotonic $\mu \in \mathcal{M}$ such that  $\mu_\succ=\rho$ if and only if   $\Delta_x \rho(y|S ) \geq 0$  whenever $x \succ y$.
\end{thm}



In words, the behavioral implication of monotonic attention is the regularity of all inferior alternatives when a better alternative is removed. The intuition is as follows. When $x \succ y$, all subsets of $S$ in which $y$ is the most preferred option cannot contain $x$; that is, $A \in S(y,\succ)$ implies $x \notin A$. It follows that removing $x$ from $S$ neither eliminates nor creates such subsets, so that $S(y,\succ) = (S \setminus \{x\})(y,\succ)$. Monotonic attention then requires that all such subsets become weakly more likely to be considered when $x$ is removed from the menu. Consequently, we must have $\rho(y|S \setminus \{x\}) \geq \rho(y|S)$. \Cref{Thm:RAM} shows that regularity for inferior alternatives is not only necessary but also sufficient for a choice rule to admit a representation under a monotonic attention rule.

In our framework, we can leverage this result directly. Suppose that $\rho$ admits an AIM-DM representation. Then the induced choice rule generated by the internal attention rule $\delta$ can be written as, for $x\in S \in \mathcal{D}$
$$
\delta_{\succ}(x|S) = \frac{\rho(x|S) - \beta \gamma_{\succ}(x|S)}{1-\beta}.
$$
Since $\delta$ is assumed to be monotonic, it follows that $\delta_{\succ}$ must satisfy the lower-contour regularity condition stated in \Cref{Thm:RAM}. This requirement allows us to further restrict the admissible range of the permeation parameter $\beta$. Applying the differencing operator yields
\begin{equation*}\tag{$*$}
\frac{\Delta_x \rho(y|S) - \beta \Delta_x \gamma_{\succ}(y|S)}{1-\beta}.
\end{equation*}

The expression $(*)$ must be nonnegative whenever $x \succ y$. Several cases arise. We begin with the case in which $\Delta_x \gamma_{\succ}(y|S) = 0$. In this situation, monotonicity imposes no additional restriction on $\beta$, since the expression is nonnegative whenever $\Delta_x \rho(y|S) \geq 0$. Intuitively, when the external presentation probability of alternative $y$ is unaffected by the removal of $x$, external stimuli exert the same influence on attention toward $y$ across the two menus. As a result, the force of monotonic attention operates directly on the observed choice rule $\rho$, yielding the standard regularity condition for alternative $y$ in the original data.

In general, depending on the sign of $\Delta_x\gammaS(y|S)$, we will obtain an upper bound or a lower bound for the attention parameter $\beta$. Applying regularity over $(*)$ and rearrange the inequality, we get
\begin{align*}
     \beta 
    \begin{dcases}
       \leq  \frac{\Delta_x \rho(y|S)}{ \Delta_x \gammaS(y|S)}  & \text{ if }  \Delta_x \gammaS(y|S) \geq0 \\
        \geq  \frac{\Delta_x \rho(y|S)}{ \Delta_x \gammaS(y|S)}  & \text{ if } \Delta_x \gammaS(y|S) <0     
    \end{dcases}
\end{align*}
Specifically, when $\Delta_x \gammaS(y|S)>\Delta_x \rho(y|S)>0$, we will claim an upper bound in $(0,1)$ for $\beta$. Intuitively, when the external stimuli implies a larger increment in choice probability of $y$ once $x$ is removed, $\beta$ cannot be too large; otherwise the internal choice rule will exhibit a decrease in choice probability of $y$ to explain the observed choice probability. Yet, as $x\succ y$, such a decrease is not allowed under a monotonic internal attention rule. By the same reasoning, if $\Delta_x \gammaS(y|S)>0>\Delta_x \rho(y|S)$, then we will rule out all $\beta\in[0,1)$ and thus eliminate the possibility of the data $\rho$ being represented by the model under the preference $\succ$ because the internal choice rule must exhibit a decrease in choice probability of $y$.

Similarly, when $\Delta_x \gammaS(y|S)<\Delta_x \rho(y|S)<0$, we will claim a lower bound in $(0,1)$ for $\beta$. The idea is that, if the choice rule $\rho$ violates regularity with respect to an inferior alternative, it must be due to the external choice rule $\gammaS$ since the internal choice rule $P_\succ$ cannot exhibit such a violation. Hence, to rationalize the data with respect to our model, one must conclude a minimal impact of $\beta$.  Notice that this was not possible in the general framework, where one can only claim a maximum impact of external stimuli.

Therefore, in general, we obtain a bound—either an upper or a lower one—whenever the sign of $\Delta_x \gamma_{\succ}(y|S)\cdot \Delta_x \rho(y|S)$ is (weakly) positive. If the sign is negative, then either the model is rejected or the resulting bound is nonbinding. We summarize these observations in \Cref{tab:Revealing Bound}.

\begin{table}
    \centering
    \renewcommand{\arraystretch}{1.2} 
    \begin{tabular}{ccccc}
        \toprule
     &   &  \multicolumn{2}{c}{$\Delta_x \rho(y|S)$} \\
        \cmidrule(lr){3-4}
        & & $ \geq 0$ &  $<0$ \\
        \midrule
        \multirow{2}{*}{$\Delta_x \gammaS(y|S)$} 
            & $\geq 0$ & Upper bound &  Rejected  \\
            & $<0$ & Non-binding bound     & Lower bound \\
        \bottomrule
    \end{tabular}
    \caption{Revealing the Bound for Permeation Paramter}
    \label{tab:Revealing Bound}
\end{table}

To formally state our  result, we define the following two sets $U_\succ$ and $L_\succ$. 
\begin{align*}
 U_{\succ}   :=\{(x,y,S) \in X^2 \times \mathcal{D}: x,y \in S, x\succ y, S\setminus x \in \mathcal{D}\text{ and } \Delta_x \gammaS(y|S) \geq 0 \} \\
 L_{\succ}   :=\{(x,y,S) \in X^2 \times \mathcal{D}: x,y \in S, x\succ y, S\setminus x \in \mathcal{D} \text{ and } \Delta_x \gammaS(y|S)<0 \}  
\end{align*}
Thus, $U_\succ$ consists of all tuples $(x,y,S)$ from which we can pin down an upper bound for $\beta$; $L_\succ$ consists of all tuples $(x,y,S)$ from which we can pin down an lower bound for $\beta$. Therefore, we denote the following two sets of inequalities.

\begin{equation}\tag{AIM-DM-UB}
\beta \le  
\frac{\Delta_x\rho(y|S)}{\Delta_x\gammaS(y|S)} \ \ \text{for all } (x,y,S) \in U_\succ 
\end{equation}

\begin{equation}\tag{AIM-DM-LB}
\beta \ge 
\frac{\Delta_x\rho(y|S)}{\Delta_x\gammaS(y|S)} \ \  \text{for all \ } (x,y,S) \in L_\succ  
\end{equation}

For convenience, we define $\beta^{DM}_\succ(x,y,S):=\dfrac{\Delta_x\rho(y|S)}{\Delta_x\gammaS(y|S)}$. We are now ready to state our main result on identifying bounds for the attention parameter under the assumption of monotonic attention.

\begin{thm}[Revealed Bound of AIM-DM]\label{Thm:RevealedBound_AIM_DM}
Given a dataset $\mathcal{D}$,
$$\BB_\succ^{DM}(\rho,\gamma) \subseteq \{\beta \in[0,1): \beta \text{ satisfies (AIM-D-UB), (AIM-DM-UB) and (AIM-DM-LB)} \}.$$
Moreover, if $\mathcal{D} = 2^X \setminus \{\emptyset\}$, the set inclusion holds with equality.
\end{thm}

\Cref{Thm:RevealedBound_AIM_DM} provides a complete characterization of the revealed range when 
$\mathcal{D}=2^X\setminus\{\emptyset\}$. The result can be expressed equivalently as the intersection of the interval generated by the monotonicity restrictions and the interval generated by the perfect-displacement upper bound:
\begin{align*}
    \BB_\succ^{DM}(\rho,\gamma) 
    &\subseteq  
    \left[
    \max_{(x,y,S)\in L_\succ} \beta^{DM}_\succ(x,y,S),
    \min_{(x,y,S)\in U_\succ} \beta^{DM}_\succ(x,y,S)
    \right] 
    \cap 
    \left[
    0,
    \min_{(x,S)} \beta^{D}_\succ(x,S)
    \right].
\end{align*}
Equivalently,
\begin{align*}
    \BB_\succ^{DM}(\rho,\gamma) 
    &\subseteq  
    \left[
    \max \left\{ \max_{(x,y,S)\in L_\succ} \beta^{DM}_\succ(x,y,S),0 \right\},
    \min \left\{
    \min_{(x,y,S)\in U_\succ} \beta^{DM}_\succ(x,y,S),
    \min_{(x,S)} \beta^{D}_\succ(x,S)
    \right\}
    \right].
\end{align*}
The inclusion becomes an equality when $\mathcal{D}=2^X\setminus\{\emptyset\}$. For limited data, a complete characterization can also be obtained; we provide this extension in \Cref{App:LimitedData}.

Finally, we turn to the preference implications of the model. An immediate implication of 
\Cref{Thm:RevealedBound_AIM_DM} is that, for a fixed preference relation $\succ$, the model admits a representation if and only if the corresponding upper- and lower-bound inequalities in $\beta$ have a nonempty intersection. Since multiple preference orderings may be consistent with the same choice data, we adopt a conservative notion of revealed preference, following \citet{Masatlioglu2012RevealedAttention}. We say that $x$ is revealed to be preferred to $y$ if every AIM-DM representation of $(\rho,\gamma)$ ranks $x$ above $y$. This definition can be equivalently stated in terms of the emptiness of the revealed interval for any preference relation that ranks $y$ above $x$, which we state in the following remark.

\begin{rmk}[Revealed Preference]\label{Coro:revealedPrefernece}
Let $(\rho,\gamma)$ admit an AIM-DM representation. Then $x$ is revealed to be preferred to $y$ if and only if, for every preference relation $\succ'$ such that $y\succ' x$, 
\[
B^{DM}_{\succ'}(\rho,\gamma)=\emptyset.
\]
\end{rmk}

Based on \Cref{tab:Revealing Bound} and our previous deviation, there are two ways in which a preference ordering can be rejected. One is the conflicting signs of $\Delta_x\gamma_\succ(y|S)$ and $\Delta_x \rho(y|S)$ given one tuple $(x,y,S)$, and the other one is an empty intersection between the upper bounds (AIM-DM-UB or AIM-D) and the lower bound (AIM-DM-LB). Therefore, one can examine how these three contribute to the revealed preference marked in \Cref{Coro:revealedPrefernece}.

\begin{prop}[Revealed Preference Decomposition]\label{Prop:Revealed Preference}
Let $(\rho,\gamma)$ admit an AIM-DM representation. Then, $x$ is revealed to be preferred to $y$ if, for every preference relation $\succ'$ such that $y\succ' x$, at least one of the following conditions holds:
\begin{enumerate}
    \item (\emph{Sign Conflict: DM}) There exists $(w,z,S)\in U_{\succ'}$ such that
    \[
        \Delta_w \rho(z|S)<0.
    \]
    \item (\emph{Bound Conflict: DM}) There exist $(w,z,S)\in L_{\succ'}$ and $(w',z',S')\in U_{\succ'}$ such that
    \[
        \beta^{DM}_{\succ'}(w,z,S) > \beta^{DM}_{\succ'}(w',z',S') \geq 0.
    \]
    \item (\emph{Bound Conflict: DM \& D}) There exist $(w,z,S)\in L_{\succ'}$ and $(w',S')$ such that
    \[
        \beta^{DM}_{\succ'}(w,z,S) > \beta^{D}_{\succ'}(w',S').
    \]
\end{enumerate}
where the converse is also true if $\mathcal{D} = 2^X \setminus \{\emptyset\}$.
\end{prop}

For a preference $\succ'$ satisfying $y\succ' x$, these rejection channels need not be mutually exclusive: more than one condition may rule out the same preference ordering. In the simulation section, we examine how often each channel contributes to revealed preference.

\section{Simulation}\label{Sect:Simulation}

In this section, we present simulation results with two primary objectives: to examine how the proposed bounds narrow the range of admissible values for the permeation parameter $\beta$, and to evaluate the extent to which the procedure recovers the underlying preference ordering in the AIM-DM. To achieve this, we simulate choice data based on the AIM-DM. Since the choice data are simulated directly from the AIM-DM, we can compare the bounds derived solely from the AIM-D in \Cref{prop:BoundSpecial} with those incorporating the additional restrictions from \Cref{Thm:RevealedBound_AIM_DM}. As discussed below, we also allow for limited data so that $\mathcal{D}\neq 2^X$. Hence, the bounds and revealed preferences could be further improved by using the results from \Cref{App:LimitedData}. Here, for illustrative purposes and simplicity, we only present the results using the bounds from \Cref{Thm:RevealedBound_AIM_DM} and \Cref{prop:BoundSpecial}.

To simulate a choice rule $\rho$, we first need to simulate one external stimulus $\gamma$ and one internal attention rule $\delta$. For $\gamma$, we draw each $\gamma(A|S)$ independently from $\mathrm{U}(0,1)$ and normalize across all nonempty subsets of $S$ so that $\sum_{\emptyset \neq A \subseteq S}\gamma(A|S)=1.$ 

For $\delta$, since the AIM-DM requires internal monotonicity, drawing weights independently for each menu is insufficient. To increase the generality of the results with respect to the randomness used to generate the attention rule, we use six different algorithms, described below, each of which guarantees monotonicity by construction.

The first two algorithms follow the random logit attention model of \citet{Brady2016Menu-DependentFeasibility}. A single weight $w_A$ is drawn once for every nonempty subset $A$ of the grand set $X$, and
\[
  \mu(A |S) \;=\; \frac{w_A}{\sum_{B \subseteq S,\, B \neq \emptyset} w_B}.
\]
Monotonicity holds automatically since \citet{Brady2016Menu-DependentFeasibility} is a special case of \citet{Cattaneo2020AModel}. To see this, enlarging the menu increases the denominator while $w_A$ remains unchanged. Algorithm~1 uses $w_A \sim \mathrm{Exp}(1)$; Algorithm~2 uses $w_A \sim \mathrm{U}(0,1)$.

Algorithms~3--6 use a ``bottom-up'' construction in which menus are processed in ascending order of cardinality. For each proper subset $A \subset S$, define the binding ceiling
\[
  \bar{u}_A \;=\; \min\bigl\{\mu(A |S') : S' \subset S,\; A \subseteq S' \in \mathcal{D}\bigr\},
\]
set to $1$ if no sub-menu exists. Drawing each proper-subset weight from a distribution supported on $[0, \bar{u}_A]$ directly enforces monotonicity by construction.

These four algorithms are summarized in Table~\ref{tab:bottomup}. The row is the draw distribution: either $W_A \sim \mathrm{U}(0,\bar{u}_A)$ (\emph{Uniform}) or $W_A \sim \bar{u}_A \cdot \mathrm{Beta}(2,1)$ (\emph{Skewed}), where the latter concentrates mass near the ceiling and tends to assign higher attention to proper subsets. The column governs what happens when the total proper-subset mass exceeds one. When $\sum_{A \subset S} W_A \leq 1$, all four algorithms agree: the residual $1 - \sum_{A \subset S} W_A$ is assigned to $\mu(S |S)$. When the sum exceeds one, the \emph{Full-Attention Positive} rule adds a fresh draw $W_S \sim \mathrm{U}(0,1)$ and renormalizes the full weight vector, so $\mu(S |S) > 0$; the \emph{Full-Attention Zero} rule renormalizes only the proper-subset weights and sets $\mu(S |S) = 0$.

\begin{table}[ht]
\centering

\begin{tabular}{lcc}
\hline
& \multicolumn{2}{c}{When $\sum_{A \subset S} W_A > 1$} \\
\cmidrule(lr){2-3}
Draw distribution & Full-Attention Positive & Full-Attention Zero \\
\hline
Uniform:\; $W_A \sim \mathrm{U}(0,\,\bar{u}_A)$
  & Algorithm 3 & Algorithm 4 \\[4pt]
Skewed:\; $W_A \sim \bar{u}_A \cdot \mathrm{Beta}(2,1)$
  & Algorithm 5 & Algorithm 6 \\
\hline
\end{tabular}
\caption{Algorithms 3--6: bottom-up construction.}
\label{tab:bottomup}
\end{table}

The simulation proceeds as follows. The grand set consists of $N=5$ alternatives. For each observation $i$, we simulate one external stimulus $\gamma_i$ and one internal attention rule $\delta_i$ using one of the six algorithms; both are defined over all nonempty subsets of $X$. To allow for limited data, we randomly select a domain $\mathcal{D}_i \subseteq 2^X \setminus \{\emptyset\}$ by drawing a random integer $c$ uniformly from $\{1,\ldots,2^N - 1\}$ and then sampling $c$ subsets without replacement. We generate $66{,}000$ observations in total, allocated equally across the six algorithms. Within each algorithm, the observations are divided equally into $11$ groups, one for each value $\beta \in \{0, 0.1, \ldots, 1\}$, yielding $1{,}000$ independent simulations per $(\text{algorithm},\beta)$ cell. Choice probabilities are then computed according to the AIM-DM, yielding the grand dataset $\{(\rho_i, \gamma_i)\}_{i=1}^{66{,}000}$, where $\rho_i$ and $\gamma_i$ are defined on $\mathcal{D}_i$.

\subsection{Revealed permeation}

In this subsection, we examine the difference in predictive power between the two
models. Using the $66{,}000$ pairs of choice data $\{(\rho_i,\gamma_i)\}_{i=1}^{66{,}000}$,
we compute the bound sets $\{\textbf{B}^{DM}(\rho_i,\gamma_i)\}_{i=1}^{66{,}000}$ and
$\{\textbf{B}^{D}(\rho_i,\gamma_i)\}_{i=1}^{66{,}000}$ based on \Cref{prop:BoundSpecial} and
\Cref{Thm:RevealedBound_AIM_DM}. Here, $\textbf{B}^{D}$ denotes the revealed set
under AIM-D, obtained by applying the inequality (AIM-D-UB) from
\Cref{prop:BoundSpecial}. Since the six algorithms generate qualitatively similar results, we pool the
observations across all algorithms.

As discussed, the bounds under AIM-D and AIM-DM differ qualitatively. Under AIM-DM,
one can obtain not only an upper bound but also a lower bound on $\beta$. We examine
the averages of the upper and lower bounds from each model. Since the raw average of a
bound depends on the underlying attention parameters, we condition on the true
permeation parameter $\beta$ throughout.

\begin{figure}[!htbp]
    \centering
    \includegraphics[width=0.6\linewidth]{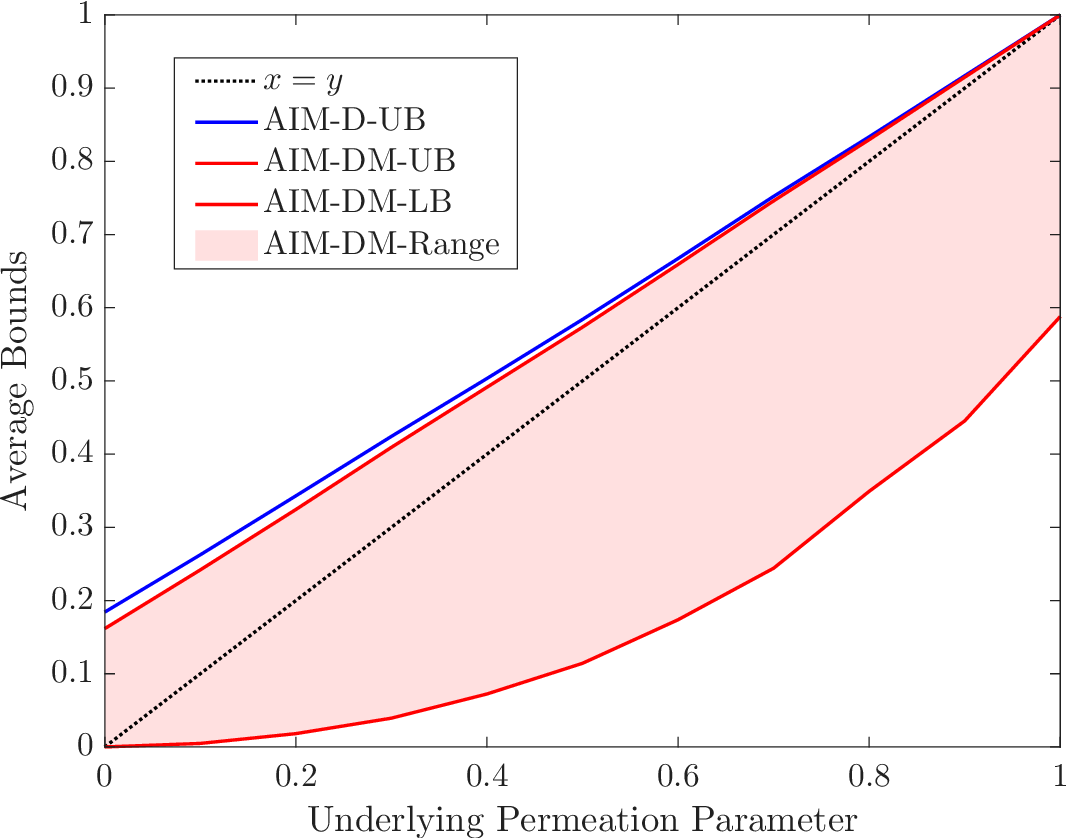}
    \caption{Average bounds as a function of the underlying permeation parameter.}
    \label{Fig; Average Bound}
\end{figure}

In \Cref{Fig; Average Bound}, the $x$-axis represents the underlying permeation
parameter $\beta \in \{0,0.1,\ldots,1\}$. For each value of $\beta$, there are
$6{,}000$ pairs of choice probabilities and attention parameters. The 45-degree
dotted diagonal is the $x=y$ reference line. The blue solid line plots the average
upper bound over the $6{,}000$ observations for each value of $\beta$ under the
general AIM-D model. The red solid lines plot the average maximum and minimum of the
set $\textbf{B}^{DM}(\rho_i,\gamma_i)$, respectively.\footnote{The set
$\textbf{B}^{DM}(\rho_i,\gamma_i)$ may be a disjoint union of multiple intervals, so
only the overall maximum and minimum of the set are reported.}

As the figure shows, the upper bound from AIM-D is revised downward  once the
monotonicity restriction is imposed (upper red solid line; AIM-DM-UB). Moreover,
unlike the AIM-D model, which yields no lower bound, the AIM-DM specification also
provides a positive lower bound (lower red solid line; AIM-DM-LB) that moves closely
with the true value of $\beta$, thereby narrowing the identified set for the
permeation parameter. One caveat is that, although the shaded region depicts the
range of AIM-DM bounds, the revealed range in each simulation can contain gaps
because $\textbf{B}^{DM}(\rho_i,\gamma_i)$ may be a disjoint union of several
intervals.

\subsection{Revealed preference}

In this subsection, we examine how much preference information can be recovered from
the simulated data. Since the data are generated from an AIM-DM, it is guaranteed that
\[
  \bigcup_{\succ \in \mathcal{P}} \textbf{B}^{DM}_\succ(\rho_i,\gamma_i) \neq \emptyset
\]
for every $i$, because the true underlying preference always yields a non-empty
revealed bound. What is less clear, however, is how many preference orderings remain
consistent with the data and how many pairwise revealed preferences can be
identified. Building on \Cref{Prop:Revealed Preference}, our simulation allows us to
address both questions.

\Cref{Prop:Revealed Preference} decomposes revealed preference into three conditions.
To conclude that $x$ is preferred to $y$, one must reject every preference ordering
that places $y$ above $x$. Our simulation further examines how much identifying power
each condition has on its own. Specifically, we ask how many pairwise revealed
preferences can be established when only a single condition---either (Sign Conflict:
DM), (Bound Conflict: DM), or (Bound Conflict: DM\,\&\,D)---is used in isolation.
The results are presented in \Cref{Fig: RevealedPreference}. As before, we condition
on the underlying permeation parameter and pool across all six algorithms, since we
find no significant qualitative differences across them.

\begin{figure}[!htbp]
    \centering
    \includegraphics[width=0.6\linewidth]{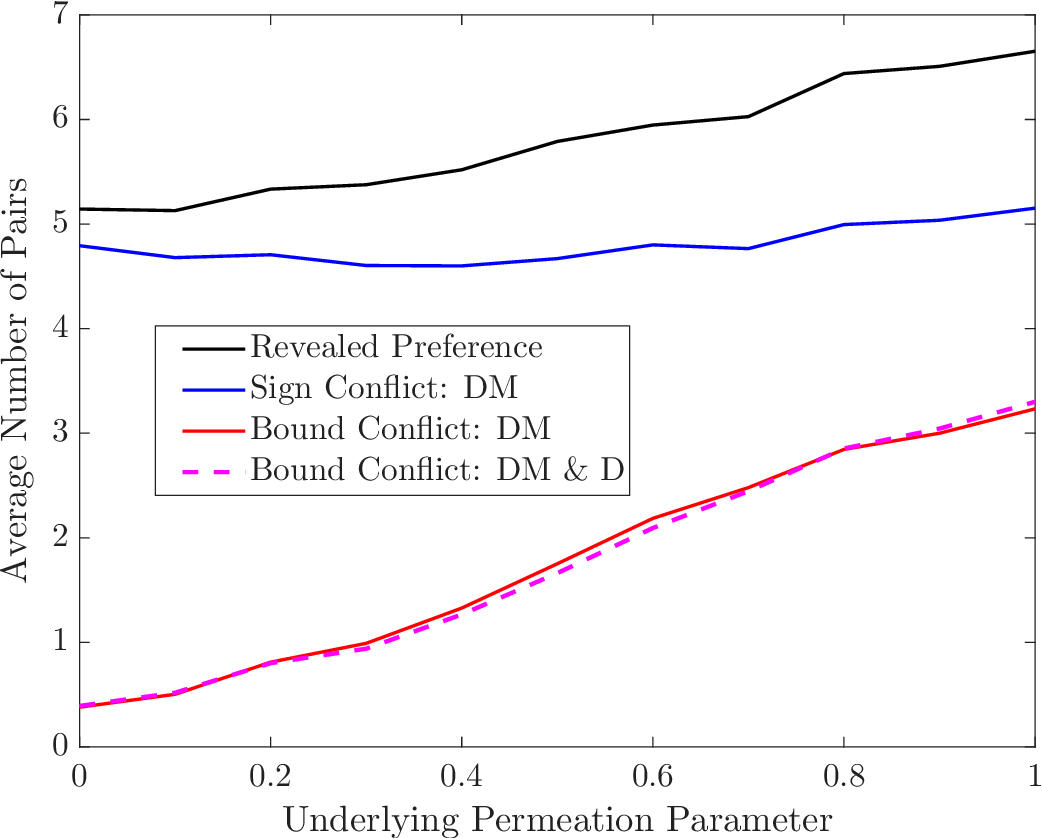}
    \caption{Average pairs of revealed preference as a function of the underlying permeation parameter.}
    \label{Fig: RevealedPreference}
\end{figure}

Since the grand set has $N=5$ alternatives, the maximum number of pairwise revealed
preferences is $10$ (because $\binom{5}{2}=10$). The black solid line shows that, on
average, between $5$ and $7$ pairs are revealed, indicating that more than half of
the preference relation can be pinned down in this simulation. The count is increasing
in $\beta$, reflecting that revealed preference becomes more powerful as the external
stimulus is more likely to enter the consideration set.

For the single-condition results, the (Sign Conflict: DM) condition, shown by the
blue solid line, is fairly stable across different values of $\beta$. In contrast,
the two bound-conflict conditions increase with $\beta$: the red solid line
corresponds to (Bound Conflict: DM), and the magenta dotted line corresponds to
(Bound Conflict: DM\,\&\,D). This pattern is natural. As $\beta$ increases, the lower
bound under AIM-DM also increases, as shown in \Cref{Fig; Average Bound}. A higher
lower bound is more likely to exceed an upper bound and therefore more likely to
reject a preference ordering. The two bound-conflict conditions produce similar
counts, which is consistent with \Cref{Fig; Average Bound}, where the AIM-DM and
AIM-D upper bounds are close to each other.

\section{Conclusion}

This paper develops a revealed-preference framework for studying how external stimuli affect attention and choice. We introduce the Attentional Interference Model, which captures two forces in consideration-set formation: proactive interference, which limits the entry of external stimuli, and retroactive interference, which allows external stimuli to displace internally generated consideration. These forces are summarized by the permeation and displacement parameters.

Our main results characterize what can be inferred about these parameters from observed choice data and observable presentation frequencies. In the unrestricted model, zero permeation can never be ruled out, but the data still impose sharp upper bounds on the extent to which external stimuli can affect attention. For the benchmark cases of Perfect Displacement, Perfect Retention, and Unified Attentional Interference, we obtain explicit permeation bounds and show that they are ordered, with Perfect Displacement yielding the most conservative bound.

We then show that stronger conclusions follow when internal attention satisfies monotonicity. Under this restriction, the model yields both upper and lower bounds on permeation, becomes falsifiable, and generates revealed-preference implications. In particular, regularity violations in observed choice can be attributed to external stimuli, allowing the analyst to infer a minimum degree of attentional permeation. The simulations illustrate these points. Imposing monotonic internal attention narrows the revealed set for the permeation parameter, produces lower bounds that track the true value of permeation, and recovers a substantial portion of the underlying preference ordering.

Overall, the analysis highlights a trade-off that is central to revealed attention models. With unrestricted attention, external stimuli can be bounded only from above, since internal attention is flexible enough to rationalize any choice pattern. With additional behavioral structure, however, choice data can reveal not only how large the effect of external stimuli could be, but also how large it must be. This distinction is relevant for applications in marketing, platform design, and policy environments in which external cues are used to guide attention without restricting choice. Future work may extend the framework to richer forms of stimulus heterogeneity, dynamic exposure, or environments in which external stimuli affect both attention and preferences. Such extensions would further clarify how observational choice data can be used to distinguish intrinsic decision making from the attentional effects of the environments in which choices are made.

\newpage

\appendix
\appendixpage

In \Cref{App:Proofs}, we provide the proofs for the results in the paper. In \Cref{App:LimitedData}, we provide a full characterization for the case of AIM-DM where $\mathcal{D} \neq 2^X \setminus \emptyset$.

\section{Proofs}\label{App:Proofs}

\subsection{Proof of \Cref{Thm:RevealAttInt}}

We first provide the following lemma, which will prove to be useful for the proof.

\begin{lemma}[Multiplicative Seperation over Cumulative Choice]
Let $\mu$ be $\alpha$-$\beta$ interfered by $\gamma$ with internal attention rule $\delta$. Then, for any $\succ$,
$$\Lmu= (1-\beta)\Ldelta+\beta \left(\alpha \Lgamma +(1-\alpha)\Lmu \Lgamma \right)$$
\end{lemma}

\begin{proof} Let $L_\succ(x):=\{y\in X: x \succ y\}$. Then,
    \begin{align*}
    \Lmu(x|S)& \equiv (1-\beta)\!\!\!\!\!\sum_{A\subseteq L_\succ(x)\cap S}\!\!\!\!\!\delta(A|S)+\beta \left( \alpha\!\!\!\!\!\sum_{A\subseteq L_\succ(x)\cap S}\!\!\!\!\!\gamma(A|S)+ (1-\alpha)\!\!\!\!\!\sum_{(B,C): B\cup C\subseteq L_\succ(x)\cap S}\!\!\!\!\!\delta(B|S)\gamma(C|S) \right)\\
    &=(1-\beta)\Ldelta(x|S)+\beta\left( \alpha \Lgamma(x|S)+ (1-\alpha)\left(\sum_{B\subseteq L_\succ(x)\cap S}\!\!\!\!\!\delta(B|S)\right)\left(\sum_{C\subseteq L_\succ(x)\cap S}\!\!\!\!\!\gamma(C|S)\right) \right)\\
    &=(1-\beta)\Ldelta(x|S)+\alpha\beta \Lgamma(x|S)+\beta (1-\alpha) \Ldelta(x|S)\Lgamma(x|S).
\end{align*}
where the second equality follows from the fact $B\cup C\subseteq L_\succ(x)\cap S\Leftrightarrow[B\subseteq L_\succ(x)\cap S\ \land\ C\subseteq L_\succ(x)\cap S]$.
\end{proof}

This lemma demonstrates how the cumulative representation enables a succinct  representation of the  of the interfered attention rule $\mu$. In fact, this rearrangement can lead to a full characterization in terms of the revealed arrange for the AIM model. Note that in \Cref{rmk:any_choice_go}, 

Given the lemma, we can rearrange the terms in the expression, so that we have 

$$\Ldelta=\frac{\Lmu-\alpha\beta \Lgamma}{(1-\beta)+\beta(1-\alpha)\Lgamma}$$

To ensure $\Ldelta$ corresponds to a well-defined cumulative choice rule, it is sufficient that to guarantee that 
\begin{enumerate}[label=(\roman*)]
    \item $\Ldelta(w_S|S)\geq 0$ for all $S$;
    \item $1 \geq \Ldelta(b_S|S)$ for all $S$;
    \item $\Ldelta(x|S)\geq \Ldelta(x'|S)$ for all $x,x'\in S$ such that $x\succ x'$.
\end{enumerate}

To prove theorem 1, we first prove $\BI_\succ(\rho,\gamma) \subseteq \Bigl\{(\alpha,\beta)\in \mathcal{A}:(\alpha,\beta)\text{ satisfies (AII)}\Bigr\}$.

If $(\alpha',\beta')\in\BI_\succ(\rho,\gamma)$, there exists an internal allocation rule $\delta \in \mathcal{M}$, and an attention rule $\mu \in \mathcal{M}$ such that $(\rho,\gamma)$ admits an $(\succ,\mu,\alpha',\beta',\delta)$ AIM representation. Given Lemma 1, we can rearrange the terms in the expression, so that we have

$$\sum_{A\subseteq L_\succ(x)\cap S}\delta(A|S)=\Ldelta(x|S)=\frac{\Lmu(x|S)-\alpha'\beta' \Lgamma(x|S)}{(1-\beta')+\beta'(1-\alpha')\Lgamma(x|S)}$$

Note that

$$\Lgamma(x|S)=\sum_{A\subseteq L_\succ(x)\cap S}\gamma(A|S)=\sum_{y\in L_\succ(x)\cap S}\gammaS(y|S)=1-\Ugamma(x|S)$$

Similarly, since $(\rho,\gamma)$ admits an $(\succ,\mu,\alpha',\beta',\delta)$ AIM representation, we also have

$$\Lmu(x|S)=\sum_{A\subseteq L_\succ(x)\cap S}\mu(A|S)=\sum_{y\in L_\succ(x)\cap S}\muS(y|S)=\sum_{y\in L_\succ(x)\cap S}\rho(y|S)=1-\rho(x|S)$$

Therefore, we can rewrite the cumulative internal allocation rule as

$$\Ldelta(x|S)=\frac{\Lmu(x|S)-\alpha'\beta' \Lgamma(x|S)}{(1-\beta')+\beta'(1-\alpha')\Lgamma(x|S)}=\frac{1-\Urho(x|S)-\alpha'\beta'(1-\Ugamma(x|S))}{(1-\beta')+\beta'(1-\alpha')(1-\Ugamma(x|S))}$$

For all $x \in S \setminus\{b_S\}$, define $w:=\min(\succ,\{y \in S:y \succ x\})$. Since $\delta(A|S) \geq 0$ for all $A \subseteq S$, we have

$$\frac{1-\Urho(x|S)-\alpha'\beta'(1-\Ugamma(x|S))}{(1-\beta')+\beta'(1-\alpha')(1-\Ugamma(x|S))}=\Ldelta(x|S) \leq \Ldelta(w|S)=\frac{1-\Urho(w|S)-\alpha'\beta'(1-\Ugamma(w|S))}{(1-\beta')+\beta'(1-\alpha')(1-\Ugamma(w|S))}$$

Notice that $\Urho(x|S)=\Urho(w|S)+\rho(x|S)$ and $\Ugamma(x|S)=\Ugamma(w|S)+\gammaS(x|S)$. Therefore, we can get

$$\frac{1-\Urho(x|S)-\alpha'\beta'(1-\Ugamma(x|S))}{(1-\beta')+\beta'(1-\alpha')(1-\Ugamma(x|S))} \leq \frac{1-\Urho(x|S)+\rho(x|S)-\alpha'\beta'(1-\Ugamma(x|S)+\gammaS(x|S))}{(1-\beta')+\beta'(1-\alpha')(1-\Ugamma(w|S)+\gammaS(x|S))}$$

To simplify, we omit $(x,S)$ in the inequality, and let $K=1-\alpha'\beta'+\alpha'\beta'\Ugamma$. The inequality can be rewritten as

\begin{align*}
\frac{K-\Urho}{K-\beta'\Ugamma} &\leq \frac{K-\Urho+\rho-\alpha'\beta'\gammaS}{K-\beta'\Ugamma+\beta'\gammaS-\alpha'\beta'\gammaS}\\ \Leftrightarrow (\beta'\gammaS-\rho)K &\leq \beta'\Urho\gammaS-\alpha'\beta'\Urho\gammaS+\beta'\rho\Ugamma+\alpha'(\beta')^2\gammaS\Ugamma
\\ \Leftrightarrow (\beta'\gammaS-\rho)(1-\alpha'\beta'+\alpha'\beta'\Ugamma) &\leq \beta'\Urho\gammaS-\alpha'\beta'\Urho\gammaS+\beta'\rho\Ugamma+\alpha'(\beta')^2\gammaS\Ugamma \\ \Leftrightarrow \alpha'\beta'(\rho(1-\Ugamma)+\gammaS(\Urho-\beta')) &\leq \rho(1-\beta'\Ugamma)+\beta'\gammaS(\Urho-1)
\end{align*}

This proves that (AII) is held for the case $x \in S \setminus \{b_S\}$.

For the case $x = b_S$, we have 

$$\rho(b_S|S)=\muS(b_S|S)=(1-\beta')\deltaS(b_S|S)+\beta'(\alpha'\gammaS(b_S|S)+(1-\alpha')\sum_{A \in S(b_s,\succ)}\sum_{(B,C):B\cup C=A}\gamma(B|S)\delta(C|S))$$

Notice that if $A \in S(b_S,\succ)$, then $A\cup C \in S(b_S,\succ)$ for all $C \subseteq S$. It implies that

$$\sum_{A \in S(b_s,\succ)}\sum_{(B,C):B\cup C=A}\gamma(B|S)\delta(C|S)) \geq \sum_{A \in S(b_s,\succ)}(\gamma(A|S)\sum_{C \subseteq S}\delta(C|S))=\sum_{A \in S(b_s,\succ)}\gamma(A|S)=\gammaS(b_S|S)$$

Hence, we get the inequality, which is actually (WB-b)

\begin{align*}
\rho(b_S|S)&=(1-\beta')\deltaS(b_S|S)+\beta'(\alpha'\gammaS(b_S|S)+(1-\alpha')\sum_{A \in S(b_s,\succ)}\sum_{(B,C):B\cup C=A}\gamma(B|S)\delta(C|S)) \\ &\geq \beta'(\alpha'\gammaS(b_S|S)+(1-\alpha')\gammaS(b_S|S)) = \beta'\gammaS(b_S|S)
\end{align*}

Now, we can check whether (AII) is held for the case $x = b_S$. By definition, we have $\Ugamma(b_S|S)=\rho(b_S|S)$ and $\Ugamma(b_S|S)=\gammaS(b_S|S)$, and thus (AII) reduces to $(\alpha'\beta'-1)(\rho(b_S|S)-\beta'\gammaS(b_S|S)) \leq 0$. It is held due to (WB-b), which we got above.

Therefore, $(\alpha',\beta') \in \Bigl\{(\alpha,\beta)\in \mathcal{A}:(\alpha,\beta)\text{ satisfies (AII)}\Bigr\}$.

We next prove $\Bigl\{(\alpha,\beta)\in \mathcal{A}:(\alpha,\beta)\text{ satisfies (AII)}\Bigr\} \subseteq \BI_\succ(\rho,\gamma)$

Define

\begin{align*}
     \delta(A|S):= 
    \begin{cases}
        \frac{1-\Urho(x|S)+\rho(x|S)-\alpha'\beta'(1-\Ugamma(x|S)+\gammaS(x|S))}{(1-\beta')+\beta'(1-\alpha')(1-\Ugamma(x|S)+\gammaS(x|S))}-\frac{1-\Urho(x|S)-\alpha'\beta'(1-\Ugamma(x|S))}{(1-\beta')+\beta'(1-\alpha')(1-\Ugamma(x|S))} & \text{ if }  A=\{x\}, x \in S \setminus \{b_S\} \\
         1-\frac{1-\rho(b_S|S)-\alpha'\beta'(1-\gammaS(b_S,S))}{(1-\beta')+\beta'(1-\alpha')(1-\gammaS(b_S,S))} & \text{ if } A=\{b_S\} \\
         0 & \text{ otherwise }
    \end{cases}
\end{align*}

and

$$\mu(A|S)=(1-\beta')\delta(A|S)+\beta'(\alpha'\gamma(A|S)+(1-\alpha')\sum_{(B,C):B\cup C=A}\gamma(A|S)\delta(C|S))$$ for all $A \subseteq S$

We claim that $(\rho,\gamma)$ admits an $(\succ,\mu,\alpha',\beta',\delta)$ AIM representation.

First, we prove that $\delta$ is well-defined; that is, $\delta(A|S) \geq 0$ for all $A \subseteq S$ and $\sum_{A \subseteq S}\delta(A|S)=1$.

For $x \in S \setminus \{b_S\}$, let $K=1-\alpha'\beta'+\alpha'\beta'\Ugamma(x|S)$, and we can rewrite

{\small
\begin{align*}
&\delta(\{x\}|S)\\
=&\frac{K-\Urho(x|S)+\rho(x|S)-\alpha'\beta'\gammaS(x|S)}{K-\beta'\Ugamma(x|S)+\beta'\gammaS(x|S)-\alpha'\beta'\gammaS(x|S)}-\frac{K-\Urho(x|S)}{K-\beta'\Ugamma(x|S))}\\
=&\frac{\beta'\Urho(x|S)\gammaS(x|S)-\alpha'\beta'\Urho(x|S)\gammaS(x|S)+\beta'\rho(x|S)\Ugamma(x|S)+\alpha'{\beta'}^2\gammaS(x|S)\Ugamma(x|S)-(\beta'\gammaS(x|S)-\rho(x|S))K}{(K-\beta'\Ugamma(x|S)+\beta'\gammaS(x|S)-\alpha'\beta'\gammaS(x|S))(K-\beta'\Ugamma(x|S)))}\\
=&\frac{\rho(x|S)(1-\beta'\Ugamma(x|S))+\beta'\gammaS(x|S)(\Urho(x|S)-1)-\alpha'\beta'(\rho(x|S)(1-\Ugamma(x|S))+\gammaS(x|S)(\Urho(x|S)-\beta'))}{(K-\beta'\Ugamma(x|S)+\beta'\gammaS(x|S)-\alpha'\beta'\gammaS(x|S))(K-\beta'\Ugamma(x|S)))}
\end{align*}
}

Since $(\alpha',\beta') \in \Bigl\{(\alpha,\beta)\in \mathcal{A}:(\alpha,\beta)\text{ satisfies (AII)}\Bigr\}$, we can apply (AII) to the numerator of $\delta(\{x\}|S)$ and get $\delta(\{x\}|S) \geq 0$.

Similarly, for $x = b_S$, we have

\begin{align*}
\delta(\{b_S\}|S)&=1-\frac{1-\rho(b_S|S)-\alpha'\beta'(1-\gammaS(b_S|S))}{(1-\beta')+\beta'(1-\alpha')(1-\gammaS(b_S|S))}\\
&=\frac{\rho(b_S|S)-\beta'\gammaS(b_S|S)}{1-\alpha'\beta'+\alpha'\beta'\gammaS(b_S|S)-\beta'\gammaS(b_S|S)}
\end{align*}

Since (AII) for $b_S$ implies that $\rho(b_S|S)-\beta'\gammaS(b_S|S) \geq 0$, we have $\delta(\{b_S\}|S) \geq 0$.

To see $\sum_{A \subseteq S}\delta(A|S)=1$, we label the alternatives by the preference $\succ$; that is, $b_S \succ...\succ x_{i+1} \succ x_i \succ... \succ x_2 \succ x_1=w_S$. Notice that $\Urho(x_i|S)=\Urho(x_{i+1}|S)+\rho(x_i|S)$ and $\Ugamma(x_i|S)=\Ugamma(x_{i+1}|S)+\gammaS(x_i|S)$. One can use the telescope method to get

\begin{align*}
\sum_{A \subseteq S}\delta(A|S)&=\sum_{x \in S}\delta(\{x\}|S)=\delta(\{w_S\}|S)+\delta(\{x_2\}|S)+\delta(\{x_3\}|S)+...+\delta(\{b_S\}|S)\\
&=1-\frac{1-\Urho(w_S|S)-\alpha'\beta'(1-\Ugamma(w_S|S))}{(1-\beta')+\beta'(1-\alpha')(1-\Ugamma(w_S|S))}=1-\frac{1-1-\alpha'\beta'(1-1)}{(1-\beta')+\beta'(1-\alpha')(1-1)}\\
&=1
\end{align*}

We have shown that $\delta$ is well-defined above. The last step is to check whether $(\rho,\gamma)$ admits an $(\succ,\mu,\alpha',\beta',\delta)$ AIM representation.

Notice that if we use the label above, we have $\Ldelta(x_i|S)=\frac{1-\Urho(x_i|S)-\alpha'\beta'(1-\Ugamma(x_i|S))}{(1-\beta')+\beta'(1-\alpha')(1-\Ugamma(x_i|S))}$ by the telescope method.

By Lemma 1, for all $x_i \in S \setminus \{b_S\}$, we have
\begin{align*}
\muS(x_i|S)&=\Lmu(x_{i+1}|S)-\Lmu(x_i|S)\\
&=((1-\beta')\Ldelta(x_{i+1}|S)+\beta'(\alpha'\Lgamma(x_{i+1}|S)+(1-\alpha')\Ldelta(x_{i+1}|S)\Lgamma(x_{i+1}|S))) \\
  &\ \ \ \ \  \ \ \ \ \         -((1-\beta')\Ldelta(x_i|S)+\beta'(\alpha'\Lgamma(x_i|S)+(1-\alpha')\Ldelta(x_i|S)\Lgamma(x_i|S)))\\
&=(1-\Urho(x_{i+1}|S))-(1-\Urho(x_i|S))\\
&=\Urho(x_i|S)-\Urho(x_{i+1}|S))=\rho(x_i|S)
\end{align*}
and
\begin{align*}
\muS(b_S|S)&=1-\Lmu(b_S|S)\\
&=1-((1-\beta')\Ldelta(b_S|S)+\beta'(\alpha'\Lgamma(b_S|S)+(1-\alpha')\Ldelta(x_i|S)\Lgamma(b_S|S)))\\
&=1-(1-\Urho(b_S|S))\\
&=\Urho(b_S|S)=\rho(b_S|S)
\end{align*}

Therefore, $(\alpha',\beta') \in \BI_\succ(\rho,\gamma)$.

This completes the proof.

\subsection{Proof of \Cref{prop:BoundSpecial}}

According to \Cref{Thm:RevealAttInt}, we have (AII) below.
\begin{align*}
    \alpha\beta(\rho(1-\Ugamma)+\gammaS(\Urho-\beta)) \leq \rho(1-\beta\Ugamma)+\beta\gammaS(\Urho-1)
\end{align*}
\begin{enumerate}
\item[(i)] \emph{(Perfect Displacement)} if $\alpha=1$, (AII) becomes
\begin{align*}
&\rho(\beta-\beta\Ugamma)+\beta\gammaS(\Urho-\beta) \leq \rho(1-\beta\Ugamma)+\beta\gammaS(\Urho-1) \\ \Leftrightarrow &\beta\gammaS(1-\beta) \leq \rho(1-\beta) \\ \Leftrightarrow &\beta \leq \frac{\rho}{\gammaS}
\end{align*}

\item[(ii)] \emph{(Perfect Retention)} if $\alpha=0$, (AII) becomes
\begin{align*}
&0 \leq \rho(1-\beta\Ugamma)+\beta\gammaS(\Urho-1) \\ \Leftrightarrow &\beta(\gammaS(1-\Urho)+\Ugamma) \leq \rho \\ \Leftrightarrow &\beta \leq \frac{\rho}{\gammaS(1-\Urho)+\Ugamma}
\end{align*}

\item[(iii)] \emph{(Unified Interference)} if $\alpha=\beta$, (AII) becomes
\begin{align*}
&\rho(\beta^2-\beta^2\Ugamma)+\beta\gammaS(\beta\Urho-\beta^2) \leq \rho(1-\beta\Ugamma)+\beta\gammaS(\Urho-1) \\ \Leftrightarrow &\beta\gammaS(1-\beta^2-\Urho+\beta\Urho) \leq \rho(1-\beta^2-\beta\Ugamma+\beta^2\Ugamma) \\ \Leftrightarrow &\beta\gammaS(1-\beta)(1+\beta-\Urho) \leq \rho(1-\beta)(1+\beta-\beta\Ugamma) \\ \Leftrightarrow &\beta\gammaS(1+\beta-\Urho) \leq \rho(1+\beta-\beta\Ugamma) \\ \Leftrightarrow &\gammaS\beta^2-(\rho(1-\Ugamma)-\gammaS(1-\Urho))\beta-\rho \leq 0
\end{align*}
Let $\Phi:= \rho(1-\Ugamma)-\gammaS(1-\Urho)$, then
\begin{align*}
&\gammaS\beta^2-(\rho(1-\Ugamma)-\gammaS(1-\Urho))\beta-\rho \\ = &\gammaS\beta^2-\Phi\beta-\rho \\ = &\gammaS(\beta-\frac{\Phi+\sqrt{\Phi^2+4\rho\gammaS}}{2\gammaS}) (\beta-\frac{\Phi-\sqrt{\Phi^2+4\rho\gammaS}}{2\gammaS}) \leq 0
\end{align*}
Note that $\sqrt{\Phi^2+4\rho\gammaS} > \sqrt{\Phi^2} = |\Phi| \geq \Phi$, it implies that $\beta-\frac{\Phi-\sqrt{\Phi^2+4\rho\gammaS}}{2\gammaS} > 0$.
Hence, we have
$$\beta-\frac{\Phi+\sqrt{\Phi^2+4\rho\gammaS}}{2\gammaS} \leq 0,$$ and thus $$\beta \leq \frac{\Phi+\sqrt{\Phi^2+4\rho\gammaS}}{2\gammaS}.$$
\end{enumerate}

\subsection{Proof of \Cref{Thm:Monotonicity_alpha}}
\par\ \par

Fix preference $\succ$ and $\alpha\in[0,1]$ and $\beta \in [0,1)$. Suppose that $(\rho,\gamma)$ admits a $(\succ,\mu,\alpha,\beta,\delta)$-AIM representation. Now, take any $\alpha'<\alpha$. We aim to show that $(\rho,\gamma)$ also admits a $(\succ,\mu,\alpha',\beta,\delta)$-AIM representation. If this is correct, then the admissible set of $\beta$ under $\alpha$ is smaller than the set under $\alpha'$, for any $\alpha'<\alpha$. 
\par
By assumption, there exists an internal attention rule $\delta$ such that $\rho=\mu_\succ$ where 
$$\mu(A|S)=(1-\beta)\delta(A|S) + \beta \left( \alpha \gamma(A|S)  + (1-\alpha)\sum_{(B,C): B\cup C =A} \gamma(B|S)\delta(C|S) \right).$$
Equivalently, we have
$$\Lrho(x|S)=(1-\beta)\Ldelta(x|S)+\beta\left(\alpha\Lgamma(x|S)+(1-\alpha)\Lgamma(x|S)\Ldelta(x|S)\right).$$
We are looking for some internal attention rule $\delta'$ such that
$$\Lrho(x|S)=(1-\beta)\Ldelta'(x|S)+\beta\left(\alpha'\Lgamma(x|S)+(1-\alpha')\Lgamma(x|S)\Ldelta'(x|S)\right).$$
Now, consider the equation
$$(1-\beta)\Ldelta+\beta\left(\alpha\Lgamma+(1-\alpha)\Lgamma\Ldelta\right)=(1-\beta)\Ldelta'+\beta\left(\alpha'\Lgamma+(1-\alpha')\Lgamma\Ldelta'\right).$$
We can solve for $\Ldelta(x|S)$:
$$\Ldelta'(x|S)=\frac{(1-\beta)\Ldelta(x|S)+[(1-\alpha)\Ldelta(x|S)+(\alpha-\alpha')]\beta\Lgamma(x|S)}{1-\beta+(1-\alpha')\beta\Lgamma(x|S)}.$$
We need to argue that $\Ldelta'$ defined by this equation is a well-defined cumulative choice rule. First, since $\alpha>\alpha'$, $\Ldelta'(x|S)\geq 0$. Second, if $x^*$ is the $\succ$-best alternative in menu $S$, then 
$$\Ldelta'(x^*|S)=\frac{(1-\beta)+[(1-\alpha)+(\alpha-\alpha')]\beta}{1-\beta+(1-\alpha')\beta}=1$$
as $\Ldelta(x^*|S)=\Lgamma(x^*|S)=1$. Finally, it remains to show $\Ldelta'(x|S)\geq\Ldelta'(y,S)$ whenever $x\succ y$. To show this, it suffices to show that $\Ldelta'$ is increasing in $\Ldelta$ and $\Lgamma$. Clearly, $\Ldelta'$ is increasing in $\Ldelta$ because $\Ldelta$ appears in the numerator only with a positive coefficient. It is also increasing in $\Lgamma$ because 
\begin{align*}
&\frac{(1-\alpha)\Ldelta+(\alpha-\alpha')}{1-\alpha'}\geq\frac{(1-\beta)\Ldelta}{1-\beta}=\Ldelta\\
\Leftrightarrow\ &(1-\alpha)\Ldelta+(\alpha-\alpha')\geq(1-\alpha')\Ldelta\\
\Leftrightarrow\ &\alpha-\alpha'\geq(\alpha-\alpha')\Ldelta\\
\Leftrightarrow\ &1\geq\Ldelta.
\end{align*}
Note that we also use the assumption $\alpha>\alpha'$ in our reasoning here. Thus, $\Ldelta'$ is increasing in $\Ldelta$ and $\Lgamma$. So, for $x\succ y$, since $\Ldelta(x|S)\geq\Ldelta(y,S)$ and $\Lgamma(x|S)\geq\Lgamma(y,S)$, we have
\begin{align*}
\Ldelta'(x|S)&=\frac{(1-\beta)\Ldelta(x|S)+[(1-\alpha)\Ldelta(x|S)+(\alpha-\alpha')]\beta\Lgamma(x|S)}{1-\beta+(1-\alpha')\beta\Lgamma(x|S)}\\
&\geq\frac{(1-\beta)\Ldelta(y,S)+[(1-\alpha)\Ldelta(y,S)+(\alpha-\alpha')]\beta\Lgamma(y,S)}{1-\beta+(1-\alpha')\beta\Lgamma(y,S)}=\Ldelta'(y,S).
\end{align*}
We conclude that $\Ldelta'$ is a well-defined cumulative choice rule. We have constructed a $(\succ,\mu,\alpha',\beta,\delta)$-AIM representation given a $(\succ,\mu,\alpha,\beta,\delta)$-AIM representation.



\subsection{Proof of \Cref{Thm:RevealedBound_AIM_DM}}

Take any $\beta\in\textbf{B}^{DM}_\succ(\rho,\gamma)$. By definition, there exists a monotonic $\delta$ such that $(\rho,\gamma)$ admits an $(\succ,\mu,\beta,\delta)$-AIM-DM representation for some $\mu$ and $\delta$. By Proposition 3-(i), $(\rho,\gamma)$ satisfies (AIM-D-UB). Moreover, let $$P_\succ(x,S|\beta)=\frac{\rho(x|S)-\beta\gammaS(x|S)}{1-\beta}$$ for all $x,S$. Then $P_\succ(\cdot,\cdot|\beta)$ follows the Random Attention model. By \Cref{Thm:RAM}, $$\Delta_xP_\succ(y,S|\beta)\coloneqq P_\succ(y,S\setminus x|\beta)-P_\succ(y,S|\beta)=\frac{\Delta_x\rho(y|S)-\beta\Delta_x\gammaS(y|S)}{1-\beta}\geq 0$$ whenever $x\succ y$. Specifically, for all $(x,y,S)\in U_\succ$ (i.e., $x,y\in S$, $x\succ y$, $S\setminus x\in\mathcal{D}$, and $\Delta_x\gammaS(y|S)\geq0$), $\Delta_xP_\succ(y,S|\beta)\geq 0$ implies $\beta\leq\frac{\Delta_x\rho(y|S)}{\Delta_x\gammaS(y|S)}$. Thus, $\beta$ satisfies (AIM-DM-UB). Similarly, for all $(x,y,S)\in L_\succ$, $\Delta_xP_\succ(y,S|\beta)\geq 0$ implies $\beta\geq\frac{\Delta_x\rho(y|S)}{\Delta_x\gammaS(y|S)}$. Thus, $\beta$ satisfies (AIM-DM-LB). Consequently, the set inclusion in the statement holds. 
\par
Assume that $\mathcal{D}=2^X\setminus\{\emptyset\}$. Now, we argue that the set inclusion holds with equality. Take any $\beta\in[0,1)$ that satisfies (AIM-D-UB), (AIM-DM-UB), and (AIM-DM-LB). Define $P_\succ(x,S|\beta)$ as above. First, we verify that $P_\succ(\cdot,\cdot|\beta)$ is a random choice rule. Because $\rho(\cdot,S)$ and $\gammaS(\cdot,S)$ are both probability measures over $S$, we have $$\sum_{x\in S}P_\succ(x,S|\beta)=\frac{\sum_{x\in S}\rho(x|S)-\beta\sum_{x\in S}\gammaS(x|S)}{1-\beta}=\frac{1-\beta}{1-\beta}=1.$$ Because of (AIM-D-UB), $\rho(x|S)-\beta\gammaS(x|S)\geq 0$ for all $x$ and $S$. Hence, $P(x,S|\beta)\geq 0$ for all $x,S$. It follows that for any menu $S$, $P(\cdot,S|\beta)$ is a probability measure over $S$.
\par
Then, we apply \Cref{Thm:RAM} to argue that $P_\succ(\cdot,\cdot|\beta)$ can be represented by the Random Attention Model with preference $\succ$. Take any $x,y$ with $x\succ y$ and any menu $S$. We want to show $$\Delta_xP_\succ(y,S|\beta)=\frac{\Delta_x\rho(y|S)-\beta\Delta_x\gammaS(y|S)}{1-\beta}\geq0.$$ Consider first $\Delta_x\gammaS(y|S)\geq0$. By (AIM-DM-UB), $\beta\leq\frac{\Delta_x\rho(y|S)}{\Delta_x\gammaS(y|S)}$, implying $\Delta_x\rho(y|S)\geq\beta\Delta_x\gammaS(y|S)$. If $\Delta_x\gammaS(y|S)<0$, then by (AIM-DM-LB), we have $\beta\geq\frac{\Delta_x\rho(y|S)}{\Delta_x\gammaS(y|S)}$. We still obtain $\Delta_x\rho(y|S)\geq\beta\Delta_x\gammaS(y|S)$. Consequently, $\Delta_xP_\succ(y,S|\beta)\geq0$ whenever $x\succ y$.
\par
By \Cref{Thm:RAM}, there exists a monotonic $\delta\in\mathcal{M}$ such that $\delta_\succ=P_\succ(|\beta)$. It follows that $$\rho(x|S)=(1-\beta)\sum_{A\in S(x,\succ)}\delta(A|S)+\beta\sum_{A\in S(x,\succ)}\gamma(A|S)$$ for all $x,S$. Hence, $(\rho,\gamma)$ admits an $(\succ,\mu,\beta,\delta)$-AIM-DM representation. This shows the set equality and completes the proof.

\section{Monotonic Attention and Limited Data}\label{App:LimitedData}
In this subsection, we consider an arbitrary collection of menus $\mathcal{D}$ from which we observe random choice behaviors. The following result from \citet{Cattaneo2020AModel} characterizes random choice rules that are defined on $\mathcal{D}$ and consistent with the random attention model.\footnote{This result is stated in the supplementary document of \citet{Cattaneo2020AModel}.}
\begin{thm}
    Consider a strict preference relation $\succ$ over $X$ and a random choice rule $P$ defined on a collection of menus $\mathcal{D}$. The followings are equivalent:
    \begin{enumerate}
        \item The random choice rule $P$ is consistent with the Random Attention Model with the preference $\succ$.
        \item For any collection $\{(x_i,S_i)\}_{i=1}^m$ such that (i) $x_i$'s are distinct, (ii) $S_i\in\mathcal{D}$ for all $i$, (iii) $\{x_1,\cdots,x_m\}\subset\cap_{i=1}^mS_i$, and (iv) $y_i\succ x_i$ for all $i$ and all $y_i\in S_i\setminus\cap_{j=1}^mS_j$, we have $$\sum_{i=1}^mP(x_i|S_i)\leq 1.$$
    \end{enumerate}
\end{thm}
To see the necessity of the second statement, suppose that we were able to observe the choice distribution on the menu $\cap_{j=1}^mS_j$. If $\sum_{i=1}^mP(x_i|S_i)>1$, then we have $$\sum_{i=1}^mP(x_i|S_i)>1\geq\sum_{i=1}^mP(x_i|\cap_{j=1}^mS_j).$$ Hence, there must exist some $i\leq m$ such that $P(x_i|S_i)>P(x_i|\cap_{j=1}^mS_j)$, a violation of regularity. However, such a violation cannot occur under the random attention model when $y_i\succ x_i$ for every $y_i\in S_i\setminus\cap_{j=1}^mS_j$. Therefore, $\sum_{i=1}^mP(x_i|S_i)\leq 1$ follows.

Note that
\begin{align*}
    \sum_{i=1}^m\frac{\rho(x_i|S_i)-\beta\gammaS(x_i|S_i)}{1-\beta}\leq 1\Leftrightarrow\sum_{i=1}^m\rho(x_i|S_i)-1\leq\beta\left[\sum_{i=1}^m\gammaS(x_i|S_i)-1\right].
\end{align*}
Thus, depending on the sign of $\sum_{i=1}^m\gammaS(x_i|S_i)-1$, this inequality yields an upper bound or a lower bound for $\beta$. Let
\begin{align*}
O_\succ\coloneqq\left\{\right.\{(x_i,S_i)\}_{i=1}^m&:m\leq|X|;\ x_i\neq x_j\ \forall\ i<j\leq m;\ S_i\in\mathcal{D}\ \forall\ i\leq m;\\
&x_i\in\cap_{j=1}^mS_j\ \forall\ i\leq m;\ y_i\succ x_i\ \forall\ i\ \forall\ y_i\in S_i\setminus\cap_{j=1}^mS_j\left.\right\}
\end{align*}
As in Section 4, consider the following two sets:
$$U_\succ^*\coloneqq\left\{\{(x_i,S_i)\}_{i=1}^m\in O_\succ:\sum_{i=1}^m\gammaS(x_i|S_i)-1<0\right\}$$
and
$$L_\succ^*\coloneqq\left\{\{(x_i,S_i)\}_{i=1}^m\in O_\succ:\sum_{i=1}^m\gammaS(x_i|S_i)-1\geq0\right\}.$$\
Then, we have the following two sets of inequalities for $\beta$:
\begin{equation*}\tag{AIM-DM-UB*}
\beta \le  
\frac{\sum_{i=1}^m\rho(x_i|S_i)-1}{\sum_{i=1}^m\gammaS(x_i|S_i)-1} \ \ \text{for all } \{(x_i,S_i)\}_{i=1}^m\in U_\succ^* 
\end{equation*}

\begin{equation*}\tag{AIM-DM-LB*}
\beta \ge 
\frac{\sum_{i=1}^m\rho(x_i|S_i)-1}{\sum_{i=1}^m\gammaS(x_i|S_i)-1} \ \  \text{for all \ } \{(x_i,S_i)\}_{i=1}^m\in L_\succ^* 
\end{equation*}
Moreover, to ensure that the internal choice behavior $\frac{\rho(x_i|S_i)-\beta\gammaS(x_i|S_i)}{1-\beta}$ is a well-defined random choice rule, $\beta$ needs to satisfy 
\begin{equation*}\tag{AIM-D-UB}
\beta \leq \frac{\rho(x|S)}{\gammaS(x|S)} \ \  \text{for all \ }x\in 
S\in\mathcal{D}
\end{equation*}
We have the following result.

\begin{thm}[Revealed Bound - Monotonicity and Limited Data]\label{thm: monotonicity and limited attention}
Given a dataset $\mathcal{D}$,
$$\BB_\succ^{DM}(\rho,\gamma)=\{\beta \in[0,1): \beta \text{ satisfies (AIM-D-UB), (AIM-DM-UB*) and (AIM-DM-LB*)} \}.$$

\end{thm}


\newpage
\printbibliography

\end{document}